\documentclass[a4paper,USenglish,cleveref,autoref,thm-restate]{lipics-v2021}
\usepackage{amsmath}
\usepackage{amssymb}
\usepackage{amsthm}
\usepackage{thmtools, thm-restate}
\usepackage{graphicx}
\usepackage{caption}
\usepackage{subcaption}
\usepackage{multirow}
\usepackage{wrapfig}
\usepackage{hyperref}

\usepackage{commands-tam}
\newcommand{\sizeof}[1]{\left\vert#1\right\vert}
\newcommand{\ts}{\textsuperscript} 

\usepackage[textsize=scriptsize]{todonotes}

\usepackage[ruled, vlined, linesnumbered]{algorithm2e}

\SetCommentSty{mycommfont}

\nolinenumbers

\newcommand{\calS}{\mathcal{S}}

\def\planter/{{\texttt{planter}}}
\def\leftcomp/{{\texttt{left-comp}}}
\def\rightcomp/{{\texttt{right-comp}}}
\def\leftarm/{{\texttt{left-arm}}}
\def\rightarm/{{\texttt{right-arm}}}
\def\crashingcounter/{{\texttt{crashing-counter}}}

\title{The Impacts of Dimensionality, Diffusion, and Directedness on Intrinsic Cross-Model Simulation in Tile-Based Self-Assembly}
\titlerunning{Intrinsic Cross-Model Simulation in Tile-Based Self-Assembly}

\author{Daniel Hader}{Department of Computer Science and Computer Engineering, University of Arkansas, USA}{dhader@uark.edu}{}{Supported in part by National Science Foundation Grant CAREER-1553166.}
\author{Matthew J. Patitz}{Department of Computer Science and Computer Engineering, University of Arkansas, USA}{patitz@uark.edu}{}{Supported in part by National Science Foundation Grant CAREER-1553166.}

\authorrunning{D. Hader and M.\,J. Patitz}
\Copyright{Daniel Hader and Matthew J. Patitz}


\begin{document}
\maketitle
\begin{abstract}
    Algorithmic self-assembly occurs when components in a disorganized collection autonomously combine to form structures and, by their design and the dynamics of the system, are forced to intrinsically follow the execution of algorithms. Motivated by applications in DNA-nanotechnology, theoretical investigations in algorithmic tile-based self-assembly have blossomed into a mature theory with research strongly leveraging tools from computability theory, complexity theory, information theory, and graph theory to develop a wide range of models and to show that many are computationally universal, while also exposing a wide variety of powers and limitations of each. In addition to computational universality, the abstract Tile-Assembly Model (aTAM) was shown to be intrinsically universal (FOCS 2012), a strong notion of completeness where a single tile set is capable of simulating the full dynamics of all systems within the model; however, this result fundamentally required non-deterministic tile attachments. This was later confirmed necessary when it was shown that the class of directed aTAM systems, those in which all possible sequences of tile attachments eventually result in the same terminal assembly, is not intrinsically universal (FOCS 2016). Furthermore, it was shown that the non-cooperative aTAM, where tiles only need to match on 1 side to bind rather than 2 or more, is not intrinsically universal (SODA 2014) nor computationally universal (STOC 2017). Building on these results to further investigate the impacts of other dynamics, Hader et al. examined several tile-assembly models which varied across (1) the numbers of dimensions used, (2) restrictions imposed on the diffusion of tiles through space, and (3) whether each system is directed, and determined which models exhibited intrinsic universality (SODA 2020). Such results have shed much light on the roles of various aspects of the dynamics of tile-assembly and their effects on the universality of each model. In this paper we extend that previous work to provide direct comparisons of the various models against each other by considering intrinsic simulations between models. Our results show that in some cases, one model is strictly more powerful than another, and in others, pairs of models have mutually exclusive capabilities. This direct comparison of models helps expose the impacts of these three important aspects of self-assembling systems, and further helps to define a hierarchy of tile-assembly models analogous to the hierarchies studied in traditional models of computation. 
\end{abstract}
\clearpage

\section{Introduction}

Self-assembling systems are those in which a disorganized collection of simple components spontaneously combine to form complex, organized structures through random motion and local interactions. From the pristine, periodic arrangements formed by crystallizing atoms to the robust coordination of dividing cells in developing organisms, such systems are the source of much complexity in nature and a topic of critical importance to many fields of research. Among them is the field of DNA nanotechnology, wherein artificial DNA strands are used as structural units that self-assemble according to the dynamics of DNA base pairing, which has seen immense success over the past several decades in harnessing the power of self-assembly to create microscopic structures with incredible precision \cite{Douglas2009,ke2012three, OrigamiTiles,RothOrigami} and even perform algorithmic tasks at the nano-scale \cite{evans2014crystals,SeemanDNARobots2010,WinfreeDNARobots2010,DNARobotNature2010,SignalTilesExperimental,rothemund2004algorithmic,woods2019diverse,Zhang2017}. Because it's difficult and expensive to accurately model the chemistry of DNA, a variety of simplifying models have been proposed to facilitate the design of DNA-based self-assembling systems. Among the more popular and effective ones are tile-assembly (TA) models where components, made of several bound DNA strands exposing small unbound portions with which other components can bind, are abstractly represented as geometric tiles whose labeled sides attach to one another according to predefined affinity rules \cite{GeoTiles,Signals,Winf98}. The advantage of these models lies not only in their success as design tools, but in their similarity to existing models studied heavily in computer science such as Wang tiles and cellular automata. This similarity isn't a coincidence either; the first TA model proposed, the abstract Tile-Assembly Model (aTAM), was designed, at least in part, to show that the dynamics of DNA-based self-assembly are algorithmically universal \cite{Winf98}. Consequently, DNA nanotechnology shares a unique relationship with the theory of computation, with theorists frequently borrowing ideas from complexity, computability, and information theory to study questions regarding, among many other things, what kinds of structures can be self-assembled, the relative difficulty of assembling different shapes, and how variations in a model's dynamics affect its algorithmic power.
\begin{wrapfigure}{r}{0.5\textwidth}
    \centering
    \includegraphics[width=\linewidth]{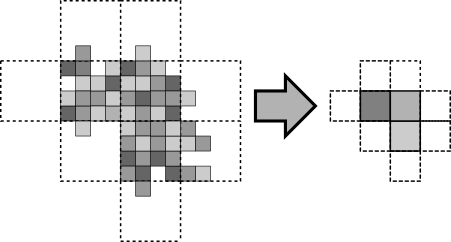}
    \caption{During an intrinsic simulation, the dynamics of individual tile attachments are simulated so that blocks of tiles in the simulating system ``look like'' individual tiles at scale.}
\end{wrapfigure}
This paper is particularly focused on that latter question. As with more conventional models of computation, we generally study such questions by proving whether one model is capable of simulating all systems of another. We have to be careful about our definition of simulation however, as it's generally straightforward to show that many TA models are capable of universal computation. Consequently, most TA models are capable of ``simulating'' all others in that they can simulate a Turing machine which can in turn simulate the other model. To learn something useful about the relative power of two TA models therefore, we have to consider the geometry of the tile-assembly dynamics. We do this by adapting a tool from the theory of cellular automata, namely \emph{intrinsic simulation}. For a simulation to be \emph{intrinsic}, we require that the simulation is not merely symbolic (i.e. how a Turing machine can simulate an aTAM system by storing an internal representation of the tiles as symbols on its tape), but rather geometric wherein blocks of tiles in the simulating system correspond to individual tiles in the simulated system and the order of tile attachments in these blocks follow those in the simulated system up to a fixed scale factor. In other words, such a simulation would appear identical to the system being simulated if we ``zoomed out'' sufficiently far. This approach is not novel to our results, in fact there is already a relatively mature theory of intrinsic simulations in tile-assembly which has resulted in a ``kind of computational complexity theory for self-assembly'' \cite{WoodsIUSurvey}. Such efforts have been instrumental in characterizing the relative power of TA models and has lead to a deeper understanding how different dynamics can be used for the same algorithmic purpose.

\subsection{Our results}

In an attempt to extend several previous results regarding intrinsic simulation, here we consider 3 specific variations of the aTAM: \emph{dimensionality}, where both 2D and 3D systems are considered, \emph{diffusion}, where tiles cannot attach in regions which have been surrounded by previously attached tiles, and \emph{directedness}, where tile attachments in a system are required to result in exactly one terminal assembly. It's important to note that these variations aren't arbitrary either. The difference between directed and undirected systems is analogous to the difference between deterministic and probabilistic algorithms and, among other things, plays a role in the study of the complexity of shape assembly \cite{SolWin07,KaoSchS08}. The diffusion restriction on the other hand is often used to make 3D tile-assembly models more ``realistic'' by limiting tile attachments to those locations in which a tile could reasonably diffuse (i.e. not in a region completely surrounded by other tiles). These variations can be introduced into the aTAM in any combination to yield 8 different models and, considering all ordered pairs of these 8 models gives rise to a table consisting of 64 entries each representing one model's ability or inability to intrinsically simulate the dynamics of another. Generally speaking, results regarding these \emph{cross-model simulations} are complex, involving intricate tile-assembly constructions and counterexamples; consequently, only a handful of these entries have been proved in past literature.

\begin{table}[ht!]
    \small
    \renewcommand{\arraystretch}{1.2}
    \centering
    \begin{tabular}{|c c|c|c|c|c|}
        \hline
        && \multicolumn{2}{c|}{\textbf{aTAM}} & \multicolumn{2}{c|}{\textbf{PaTAM}} \\
        && \textbf{all} & \textbf{dir} & \textbf{all} & \textbf{dir} \\\hline
        
        \multirow{2}{*}{\textbf{aTAM}} & \textbf{all} & \multicolumn{2}{c|}{yes* \cite{IUSA}} & \multirow{2}{*}{\shortstack{no \\(thm. \ref{thm:atam-cant-sim-patam}, obs. \ref{obs:subset-no-sim})}} & ? \\\cline{2-4}\cline{6-6}
        & \textbf{dir} & no (thm. \ref{thm:dir-cant-sim-undir}) & no* \cite{DirectedNotIU} & & ? \\\hline
        \multirow{2}{*}{\textbf{PaTAM}} & \textbf{all} & \multicolumn{2}{c|}{\multirow{2}{*}{no (thm. \ref{thm:patam-cant-sim-dir-atam}, obs. \ref{obs:subset-no-sim})}} & no* \cite{DDDIU} & yes (thm. \ref{thm:patam-can-sim-dpatam}) \\\cline{2-2}\cline{5-6}
        & \textbf{dir} & \multicolumn{2}{c|}{} & no (thm. \ref{thm:dir-cant-sim-undir}) & no* \cite{DDDIU} \\\hline
        
        \multirow{2}{*}{\textbf{3DaTAM}} & \textbf{all} & \multicolumn{2}{c|}{yes$^\dag$(obs. \ref{obs:3D-can-sim-2D-non-planar})} & ? & ? \\\cline{2-6}
        & \textbf{dir} & no (thm. \ref{thm:dir-cant-sim-undir}) & yes$^\dag$ (obs. \ref{obs:3D-can-sim-2D-non-planar}) & no (thm. \ref{thm:dir-cant-sim-undir}) & ? \\\hline
        \multirow{2}{*}{\textbf{SaTAM}} & \textbf{all} & \multicolumn{2}{c|}{yes$^\dag$ (obs. \ref{obs:3D-can-sim-2D-non-planar})} & ? & ? \\\cline{2-6}
        & \textbf{dir} & no (thm. \ref{thm:dir-cant-sim-undir}) & yes$^\dag$ (obs. \ref{obs:3D-can-sim-2D-non-planar}) & no (thm. \ref{thm:dir-cant-sim-undir}) & ? \\\hline\hline

        && \multicolumn{2}{c|}{\textbf{3DaTAM}} & \multicolumn{2}{c|}{\textbf{SaTAM}} \\
        && \textbf{all} & \textbf{dir} & \textbf{all} & \textbf{dir} \\\hline
        
        \multirow{2}{*}{\textbf{aTAM}} & \textbf{all} & \multicolumn{4}{c|}{\multirow{4}{*}{no (obs. \ref{obs:2d-cant-sim-3d})}} \\\cline{2-2}
        & \textbf{dir} & \multicolumn{4}{c|}{} \\\cline{1-2}
        \multirow{2}{*}{\textbf{PaTAM}} & \textbf{all} & \multicolumn{4}{c|}{}\\\cline{2-2}
        & \textbf{dir} & \multicolumn{4}{c|}{}\\\hline
        
        \multirow{2}{*}{\textbf{3DaTAM}} & \textbf{all} & \multicolumn{2}{c|}{yes* \cite{DDDIU}} & \multirow{2}{*}{\shortstack{no\\(thm. \ref{thm:3datam-cant-sim-satam}, obs. \ref{obs:subset-no-sim})}} & ? \\\cline{2-4}\cline{6-6}
        & \textbf{dir} & no (thm. \ref{thm:dir-cant-sim-undir}) & yes* \cite{DDDIU} & & ? \\\hline
        \multirow{2}{*}{\textbf{SaTAM}} & \textbf{all} & \multicolumn{2}{c|}{yes$^\dag$ (obs. \ref{obs:SaTAM-can-sim-3DaTAM})} & yes* \cite{DDDIU} & ? \\\cline{2-6}
        & \textbf{dir} & no (thm. \ref{thm:dir-cant-sim-undir}) & yes$^\dag$ (obs. \ref{obs:SaTAM-can-sim-3DaTAM}) & no (thm. \ref{thm:dir-cant-sim-undir}) & ? \\\hline
    \end{tabular}
    \caption{Table of our results, outlining whether the row's model can intrinsically simulate the column's model. PaTAM is the Planar aTAM, 3DaTAM the 3-dimensional aTAM, and SaTAM is the Spatial aTAM (see Section \ref{sec:model-variations} for full definitions). \emph{All} refers to the set of all systems in a model and \emph{dir} refers to the subset of directed systems. Cells marked with an asterisk (*) are existing results and those marked with a dagger ($^\dag$) are trivial observations using tile sets from existing results. All other results are novel.}
    \label{tab:results}
\end{table}

In this paper, we fill a considerable number of missing entries. Table~\ref{tab:results} lays out our results along with past results denoted by an asterisk. In it, entries are labeled to indicate whether the model in the row's header can simulate the model in the column's header. There are of course a few entries for which the answer is obvious, which we state as observations with justification rather than full theorems, but many of our results are distinctly non-trivial and some were rather unexpected. For instance, while we initially suspected that the diffusion restricted version of the aTAM (i.e. the Planar aTAM or PaTAM) was, as it's name suggests, a weaker version of the aTAM, we found that both models exhibit dynamics which cannot be simulated by the other. While the table is still missing a few entries, our contributions have brought the number of known entries up to 52 from the 16 which previously existed in published literature (8 of which were technically not explicitly stated, but were trivial observations based on the tile sets and proofs presented in \cite{DDDIU}).\footnote{It should also be noted that most of the remaining unknown entries involve simulating directed, diffusion restricted systems. While we do hope to fill these entries in the future, we suspect that their proofs will be quite complicated since simulating diffusion restricted systems is tricky and counterexamples are often harder to find in directed systems.} The rest of our paper is laid out as follows. In Section~\ref{sec:prelims}, we provide definitions of the various models and concepts used. Then in Section~\ref{sec:results} we state our results explicitly and sketch their proofs. That section is perhaps the most important since it is where we intuitively explain our results and describe how they follow from the dynamics of the model. We then provide the full proofs and technical details of our results in Section~\ref{sec:proofs}.

\section{Preliminary definitions}\label{sec:prelims}

Throughout this paper we will use $\mathbb{Z}$, $\mathbb{Z}^+$, and $\mathbb{N}$ to denote the set of integers, positive integers, and non-negative integers respectively. We will also assume $\mathbb{Z}^d$ has the additional structure of a lattice graph so that each point is a vertex and two points are adjacent (i.e.\ share an edge) exactly when their Euclidean distance is 1.

\subsection{Definition of the abstract Tile-Assembly Model}

In this section, we define the abstract Tile-Assembly Model in 2 and 3 dimensions. We will use the abbreviation \emph{aTAM} to refer to the 2D model and \emph{3DaTAM} for the 3D model. These definitions are borrowed from \cite{DDDIU} and we note that \cite{RotWin00} is a good introduction to the model for unfamiliar readers. 

Fix $d\in\{2,3\}$ to be the number of dimensions and $\Sigma$ to be some alphabet with $\Sigma^*$ its finite strings. A \emph{glue} $g\in\Sigma^*\times\mathbb{N}$ consists of a finite string \emph{label} and non-negative integer \emph{strength}. A \emph{tile type} is a tuple $t\in(\Sigma^*\times\mathbb{N})^{2d}$, thought of as a unit square or cube with a glue on each side. A \emph{tile set} is a finite set of tile types. We always assume a finite set of tile types, but allow an infinite number of copies of each tile type to occupy locations in the $\mathbb{Z}^d$ lattice, each called a \emph{tile}.

Given a tile set $T$, a \emph{configuration} is an arrangement (possibly empty) of tiles in the lattice $\mathbb{Z}^d$, i.e.\ a partial function $\alpha:\mathbb{Z}^d\dashrightarrow T$. Two adjacent tiles in a configuration \emph{interact}, or are \emph{bound} or \emph{attached}, if the glues on their abutting sides are equal (in both label and strength) and have positive strength. Each configuration $\alpha$ induces a \emph{binding graph} $B_\alpha$ whose vertices are those points occupied by tiles, with an edge of weight $s$ between two vertices if the corresponding tiles interact with strength $s$. An \emph{assembly} is a configuration whose domain (as a graph) is connected and non-empty. The \emph{shape} $S_\alpha \subseteq \mathbb{Z}^d$ of assembly $\alpha$ is the domain of $\alpha$. For some $\tau\in\mathbb{Z}^+$, an assembly $\alpha$ is \emph{$\tau$-stable} if every cut of $B_\alpha$ has weight at least $\tau$, i.e.\ a $\tau$-stable assembly cannot be split into two pieces without separating bound tiles whose shared glues have cumulative strength $\tau$. Given two assemblies $\alpha,\beta$, we say $\alpha$ is a \emph{subassembly} of $\beta$ (denoted $\alpha \sqsubseteq \beta$) if $S_\alpha \subseteq S_\beta$ and for all $p\in S_\alpha$, $\alpha(p)=\beta(p)$. 

A \emph{tile-assembly system} (TAS) is a triple $\calT=(T,\sigma,\tau)$, where $T$ is a tile set, $\sigma$ is a finite $\tau$-stable assembly called the \emph{seed assembly}, and $\tau\in\mathbb{Z}^+$ is called the \emph{binding threshold}.
Given a TAS $\calT=(T,\sigma,\tau)$ and two $\tau$-stable assemblies $\alpha$ and $\beta$, we say that $\alpha$ \emph{$\calT$-produces} $\beta$ \emph{in one step} (written $\alpha \to^{\calT}_1 \beta$) if $\alpha \sqsubseteq \beta$ and $\sizeof{S_\beta \setminus S_\alpha} = 1$.
That is, $\alpha \to^{\calT}_1 \beta$ if $\beta$ differs from $\alpha$ by the addition of a single tile.
The \emph{$\calT$-frontier} is the set $\partial^{\calT}\alpha = \bigcup_{\alpha \to^{\calT}_1 \beta} S_\beta \setminus S_\alpha$ of locations in which a tile could $\tau$-stably attach to $\alpha$.

We use $\mathcal{A}^T$ to denote the set of all assemblies of tiles in tile set $T$. Given a TAS $\calT=(T, \sigma, \tau)$, a sequence of $k\in\mathbb{Z}^+ \cup \{\infty\}$ assemblies $\alpha_0, \alpha_1, \ldots$ over $\mathcal{A}^T$ is called a \emph{$\calT$-assembly sequence} if, for all $1\le i < k$, $\alpha_{i-1} \to^{\calT}_1 \alpha_i$. The \emph{result} of an assembly sequence is the unique limiting assembly of the sequence. For finite assembly sequences, this is the final assembly; whereas for infinite assembly sequences, this is the assembly consisting of all tiles from any assembly in the sequence. We say that \emph{$\alpha$ $\calT$-produces $\beta$} (denoted $\alpha\to^{\calT} \beta$) if there is a $\calT$-assembly sequence starting with $\alpha$ whose result is $\beta$. We say $\alpha$ is \emph{$\calT$-producible} if $\sigma\to^{\calT}\alpha$ and write $\prodasm{\calT}$ to denote the set of $\calT$-producible assemblies. We say $\alpha$ is \emph{$\calT$-terminal} if $\alpha$ is $\tau$-stable and there exists no assembly which is $\calT$-producible from $\alpha$. We denote the set of $\calT$-producible and $\calT$-terminal assemblies by $\termasm{\calT}$.

When $\calT$ is clear from context, we may omit $\calT$ from the notation above.

\paragraph*{Cooperative Attachment}

Given a TAS $\calT=(T,\sigma,\tau)$, for a tile to attach to an assembly it must match glues whose cumulative strength is at least $\tau$ in order to result in a $\tau$-stable assembly. This can happen if, for instance, one of the matched glues has strength at least $\tau$, in which case any other matching glues are superfluous. Alternatively, a tile may still attach without any $\tau$-strength glues though this requires multiple glues to match whose strengths sum to at least $\tau$. We refer to such attachments as \emph{cooperative}.

\subsection{Model Variations}\label{sec:model-variations}

In this paper we consider 3 variations of the aTAM. Other than the 3D aTAM, these include \emph{directed} and \emph{diffusion restricted} versions of the models. We say that a TAS $\calT$ is \emph{directed} if $\sizeof{\termasm{\calT}}=1$, i.e.\ $\calT$ admits only a single producible terminal assembly. When we refer to a \emph{directed model} we simply mean the set of all directed systems in a model. Directed systems are desirable for self-assembly since we often want our tiles to grow into a single target shape. 

For diffusion restricted models, we note that in the aTAM it's possible for tiles to attach within a region of space which has been completely surrounded by other tiles. In 2D, we can imagine that the tiles are able to navigate around the assembly through the 3\ts{rd} dimension, but in 3D such attachments are difficult to justify. Consequently, we also consider models where such attachments are forbidden. In 2D, this restriction could model a self-assembly process on the surface of a droplet of water where surface tension prevents the components from taking advantage of the 3\ts{rd} dimension. We call the 2D diffusion restricted aTAM the \emph{Planar aTAM} or PaTAM, and we call the 3D diffusion restricted aTAM the \emph{Spatial aTAM} or SaTAM. In these models, and their directed subsets, we refer to regions which have been completely surrounded (in which no tile attachments are allowed to occur) \emph{constrained}. 
To formally model this restriction, we first note that given a finite $d$-dimensional assembly $\alpha$, the graph $\Z^d\setminus S_\alpha$ consists of a finite number of connected components, exactly one of which will be infinite in size. We say that this component graph is the \emph{outside} of $\alpha$ while the finite-sized components are \emph{constrained}. In a diffusion restricted system we only allow tile attachments on the outside of an assembly.

\subsection{Intrinsic Simulation}

First we provide a high-level definition of the notion of \emph{intrinsic simulation} which should be sufficient for understanding our results. A full technical definition follows afterward. 
For brevity, in this paper, unless explicitly stated, ``simulation'' will refer to intrinsic simulation.

\paragraph*{High-Level Description of Simulation}

Simulation of system $\calT$ by system $\calS$ occurs at a scale factor $m$, so that $m\times m$ (or $m \times m \times m$ in 3D) blocks of tiles from $\calS$, which we refer to as \emph{macrotiles}, correspond to individual tiles in $\calT$. For a given simulation, we define a \emph{macrotile representation function} $R$ which describes this mapping of macrotiles to tiles. Additionally for convenience, using $R$ we define an \emph{assembly representation function} $R^*$ which maps entire assemblies from $\calS$ to assemblies in $\calT$, essentially evaluating $R$ on each macrotile location for a given assembly in $\calS$. Note that we don't require all locations within a macrotile to contain a tile and macrotile blocks containing tiles can still be mapped to empty space under $R$.
When a tile attachment causes the representation of a macrotile location to map to a tile for the first time, we say that the attachment has caused the macrotile to \emph{resolve} and once a macrotile has resolved, any additional tile attachments within the macrotile cannot change its representation under $R$.
While we do allow macrotile locations to map to empty space, for a simulation to be valid there must be restrictions on where tiles are allowed to attach in $\calS$. 
For our notion of simulation to be useful as a metric of comparing the relative capabilities of models, we require that $\calS$ only place tiles within the macrotile regions immediately adjacent (not diagonally) to those which have already resolved, and we call such locations \emph{fuzz}. 
This allows tiles in $\calS$ to attach only in macrotiles which could potentially resolve during a valid simulation, since only the locations in $\calT$ mapped to by the fuzz locations could possibly receive tiles in $\calT$. If a class of systems $C$ can all be simulated by another class of systems $C'$ sharing a single tile set (though each may have a different seed assembly), we say that class $C'$ \emph{intrinsically simulates} $C$ with a \emph{universal tile set}. We can also say that $C'$ is \emph{intrinsically universal} (IU) for $C$. 


\paragraph*{Formal Definition of Simulation}

Now we provide formal definitions for \emph{intrinsic simulation}. The definitions here are taken from \cite{DDDIU} and specifically refer to 3D systems. Similar definitions for 2D intrinsic simulation are given in \cite{DirectedNotIU}. For simulation of a 2D system by a 3D system, we use the 3D definitions and assume that all systems in the 2D system are defined in 3D so that assemblies occupy only the $z=0$ plane.


From this point on, let $T$ be a tile set and let the scale factor be $m\in\Z^+$.
An \emph{$m$-block macrotile} over $T$ is a partial function $\alpha : \Z_m^3 \dashrightarrow T$, where $\Z_m = \{0,1,\ldots,m-1\}$.
Let $B^T_m$ be the set of all $m$-block macrotiles over $T$.
The $m$-block with no domain is said to be $\emph{empty}$.
For a general assembly $\alpha:\Z^3 \dashrightarrow T$ and $(x',y',z')\in\Z^3$, define $\alpha^m_{(x',y',z')}$ to be the $m$-block macrotile defined by $\alpha^m_{(x',y',z')}(i_x,i_y,i_z) = \alpha(mx'+i_x,my'+i_y,mz'+i_z)$ for $0 \leq i_x,i_y,i_z< m$.
For some tile set $S$, a partial function $R: B^{S}_m \dashrightarrow T$ is said to be a \emph{valid $m$-block macrotile representation} from $S$ to $T$ if for any $\alpha,\beta \in B^{S}_m$ such that $\alpha \sqsubseteq \beta$ and $\alpha \in \dom R$, then $R(\alpha) = R(\beta)$.

For a given valid $m$-block macrotile representation function $R$ from tile set~$S$ to tile set $T$, define the \emph{assembly representation function}\footnote{Note that $R^*$ is a total function since every assembly of $S$ represents \emph{some} assembly of~$T$; the functions $R$ and $\alpha$ are partial to allow undefined points to represent empty space.}  $R^*: \mathcal{A}^{S} \rightarrow \mathcal{A}^T$ such that $R^*(\alpha') = \alpha$ if and only if $\alpha(x,y,z) = R\left(\alpha'^m_{(x,y,z)}\right)$ for all $(x,y,z) \in \Z^3$.
For an assembly $\alpha' \in \mathcal{A}^{S}$ such that $R^*(\alpha') = \alpha$, $\alpha'$ is said to map \emph{cleanly} to $\alpha \in \mathcal{A}^T$ under $R^*$ if for all non empty blocks $\alpha'^m_{(x,y,z)}$, $(x,y,z)+(u_x,u_y,u_z) \in \dom(\alpha)$ for some $(u_x,u_y,u_z) \in U_3$ such that $u^2_x + u^2_y + u^2_z \le 1$, or if $\alpha'$ has at most one non-empty $m$-block $\alpha^m_{0,0}$.  In other words, $\alpha'$ may have tiles on macrotile blocks representing empty space in $\alpha$, but only if that position is adjacent to a tile in $\alpha$.  We call such growth ``around the edges'' of $\alpha'$ \emph{fuzz} and thus restrict it to be adjacent to only valid macrotiles, but not diagonally adjacent (i.e.\ we do not permit \emph{diagonal fuzz}).

In the following definitions, let $\mathcal{T} = \left(T,\sigma_T,\tau_T\right)$ be a TAS, let $\mathcal{S} = \left(S,\sigma_S,\tau_S\right)$ be a TAS, and let $R$ be an $m$-block representation function $R:B^S_m \rightarrow T$.

\begin{definition}
\label{def-equiv-prod} We say that $\mathcal{S}$ and $\mathcal{T}$ have \emph{equivalent productions} (under $R$), and we write $\mathcal{S} \Leftrightarrow \mathcal{T}$ if the following conditions hold:
\begin{enumerate}
        \item $\left\{R^*(\alpha') | \alpha' \in \prodasm{\mathcal{S}}\right\} = \prodasm{\mathcal{T}}$.
        \item $\left\{R^*(\alpha') | \alpha' \in \termasm{\mathcal{S}}\right\} = \termasm{\mathcal{T}}$.
        \item For all $\alpha'\in \prodasm{\mathcal{S}}$, $\alpha'$ maps cleanly to $R^*(\alpha')$.
\end{enumerate}
\end{definition}

\begin{definition}
\label{def-t-follows-s} We say that $\mathcal{T}$ \emph{follows} $\mathcal{S}$ (under $R$), and we write $\mathcal{T} \dashv_R \mathcal{S}$ if $\alpha' \rightarrow^\mathcal{S} \beta'$, for some $\alpha',\beta' \in \prodasm{\mathcal{S}}$, implies that $R^*(\alpha') \to^\mathcal{T} R^*(\beta')$.
\end{definition}

The next definition essentially specifies that every time $\mathcal{S}$ simulates an assembly $\alpha \in \prodasm{\mathcal{T}}$, there must be at least one valid growth path in $\mathcal{S}$ for each of the possible next steps that $\mathcal{T}$ could make from $\alpha$ which results in an assembly in $\mathcal{S}$ that maps to that next step. While this definition is unfortunately dense, it accommodates subtle situations such as where $\calS$ must ``commit to'' a subset of possible representations in $\mathcal{T}$ before being explicitly mapped, under $R$, to any one in particular.

\begin{definition}
\label{def-s-models-t} We say that $\mathcal{S}$ \emph{models} $\mathcal{T}$ (under $R$), and we write $\mathcal{S} \models_R \mathcal{T}$, if for every $\alpha \in \prodasm{\mathcal{T}}$, there exists $\Pi \subset \prodasm{\mathcal{S}}$ where $\Pi \neq \emptyset$ and  $R^*(\alpha') = \alpha$ for all $\alpha' \in \Pi$, such that, for every $\beta \in \prodasm{\mathcal{T}}$ where $\alpha \rightarrow^\mathcal{T} \beta$, (1) for every $\alpha' \in \Pi$ there exists $\beta' \in \prodasm{\mathcal{S}}$ where $R^*(\beta') = \beta$ and $\alpha' \rightarrow^\mathcal{S} \beta'$, and (2) for every $\alpha'' \in \prodasm{\mathcal{S}}$ where $\alpha'' \rightarrow^\mathcal{S} \beta'$, $\beta' \in \prodasm{\mathcal{S}}$, $R^*(\alpha'') = \alpha$, and $R^*(\beta') = \beta$, there exists $\alpha' \in \Pi$ such that $\alpha' \rightarrow^\mathcal{S} \alpha''$.
\end{definition}

\begin{definition}\label{def:s-simulates-t}
\label{def-s-simulates-t} We say that $\mathcal{S}$ \emph{intrinsically simulates} $\mathcal{T}$ (under $R$) if $\mathcal{S} \Leftrightarrow_R \mathcal{T}$ (equivalent productions), $\mathcal{T} \dashv_R \mathcal{S}$ and $\mathcal{S} \models_R \mathcal{T}$ (equivalent dynamics).
\end{definition}

\subsection{Window Movie Lemma}

In \cite{temp1notIU}, the authors proved the Window Movie Lemma, a pumping lemma of sorts for the aTAM (and its variants) which has since seen much use as a powerful tool for proving that certain tile-assembly simulations are impossible. Since it appears in several of our proofs, we first informally describe the lemma, then explicitly state it. A \emph{window} is an edge cut which partitions the lattice graph ($\mathbb{Z}^2$ in 2D or $\mathbb{Z}^3$ in 3D) into two regions. Given some window $w$ and some assembly sequence $\vec{\alpha}$ in a TAS $\calT$, a \emph{window movie} $M$ is defined to be the ordered sequence of glues presented along $w$ by tiles in $\calT$ during the assembly sequence $\vec{\alpha}$. Informally, if we think of the window $w$ as a thin pane dividing two regions of tile locations and imagine stepping through the assembly sequence $\vec{\alpha}$ one tile attachment at a time, $M$ is constructed by recording the glues which appear on the surface of the pane and their relative order. More formally, a \emph{window movie} is the sequence $M^{\vec{\alpha}}_w = \{(v_i, g_i)\}$ of pairs of grid graph vertices $v_i$ and glues $g_i$, given by order of appearance of the glues along window $w$ during $\vec{\alpha}$. Furthermore, if $k$ glues appear along $w$ during the same assembly step in $\vec{\alpha}$, then these glues appear contiguously and are listed in lexicographical order of the unit vectors describing their orientation in $M^{\vec{\alpha}}_w$.

\begin{figure}
    \centering
    \includegraphics[width=0.6\textwidth]{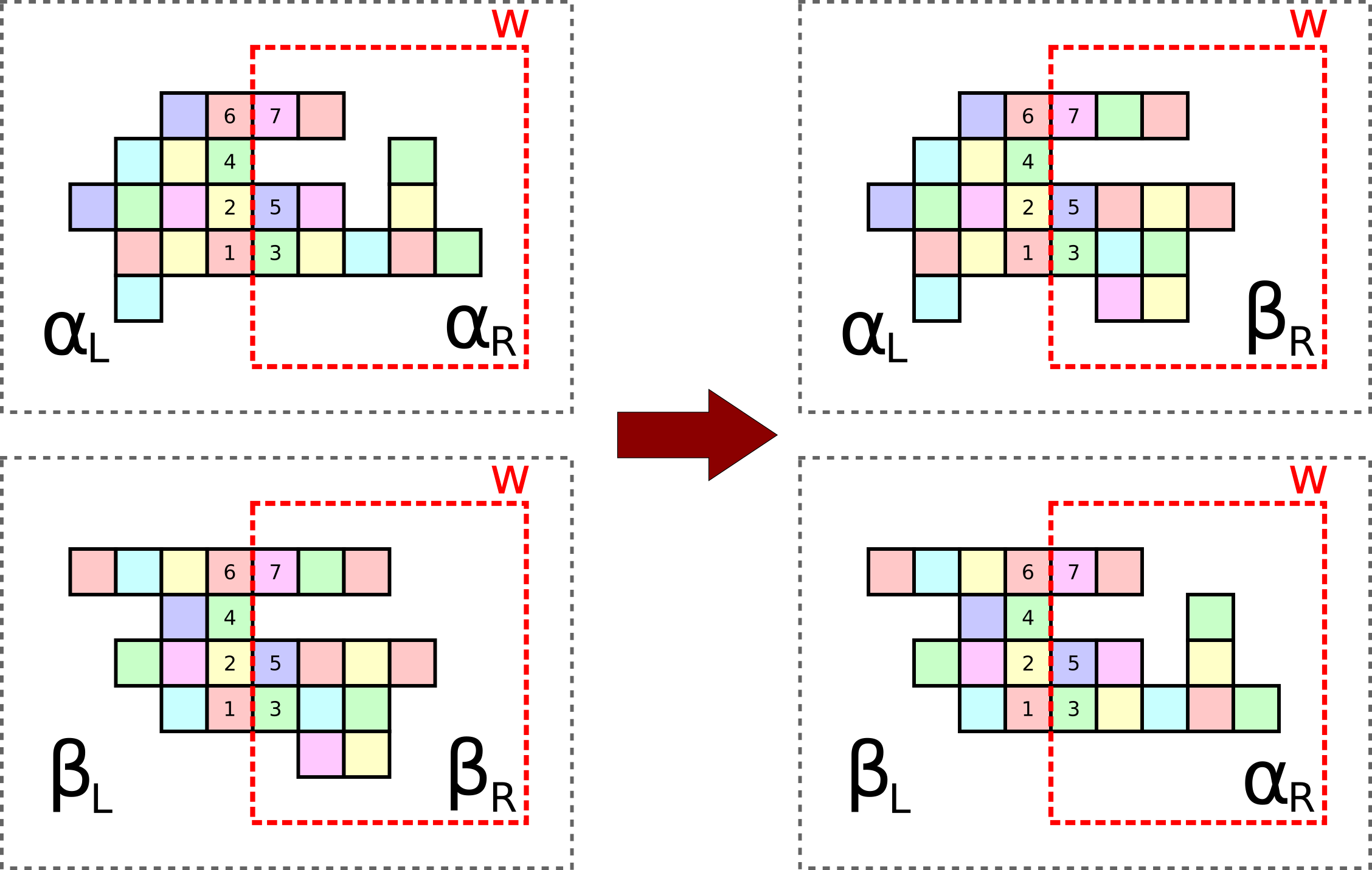}
    \caption{An illustration of the window movie lemma. On the left are two producible assemblies $\alpha = \alpha_L \cup \alpha_R$ and $\beta = \beta_L \cup \beta_R$ made from the same tile set, which are each divided into two subassemblies by the window $w$. For both assemblies, the window $w$ has the same window movie, i.e.\ the order in which tiles present glues along the window, depicted by numbers on the tiles describing the relative order in which they attached. Since all growth within the windowed regions depends only on the glues presented along the window, we can splice these assemblies to get $\alpha_L \cup \beta_R$ or $\beta_L \cup \alpha_R$ (illustrated on the right). The window movie lemma then guarantees that both of these assemblies are producible.}
    \label{fig:window-movie-lemma}
\end{figure}

Informally, the Window Movie Lemma states that any tile attachments that occur within the region bounded by a window are possible in a region bounded by the same window (up to translation) with an identical window movie. This allows us to splice assembly sequences together and, consequently, pump a sequence of tile attachments so long as we can ensure the existence of identical window movies. \Cref{fig:window-movie-lemma-pumping} illustrates how the Window Movie Lemma can be used to pump growth.



\paragraph*{Window Movie Lemma.} Let $\vec{\alpha} = \{\alpha_i\}$ and $\vec{\beta}=\{\beta_i\}$ be assembly sequences in TAS $\calT$ and let $\alpha, \beta$ be the result assemblies of each respectively. Let $w$ be a window that partitions $\alpha$ into two configurations $\alpha_L$ and $\alpha_R$ and let $w'=w+\vec{c}$ be a translation of $w$ that partitions $\beta$ into two configurations $\beta_L$ and $\beta_R$ (with $\alpha_L$ and $\beta_L$ being the configurations containing their respective seed tiles). Furthermore define $M^{\vec{\alpha}}_w$ and $M^{\vec{\beta}}_w$ to be the window movies for $\vec{\alpha},w$ and $\vec{\beta},w'$ respectively. Then if $M^{\vec{\alpha}}_w = M^{\vec{\beta}}_w$, the assemblies $\alpha_L\cup\beta'_R$ and $\beta'_L\cup\alpha_R$ (where $\beta'_L = \beta_L - \vec{c}$ and $\beta'_R = \beta_R - \vec{c}$) are also producible.

\begin{figure}
    \centering
    \includegraphics[width=0.6\textwidth]{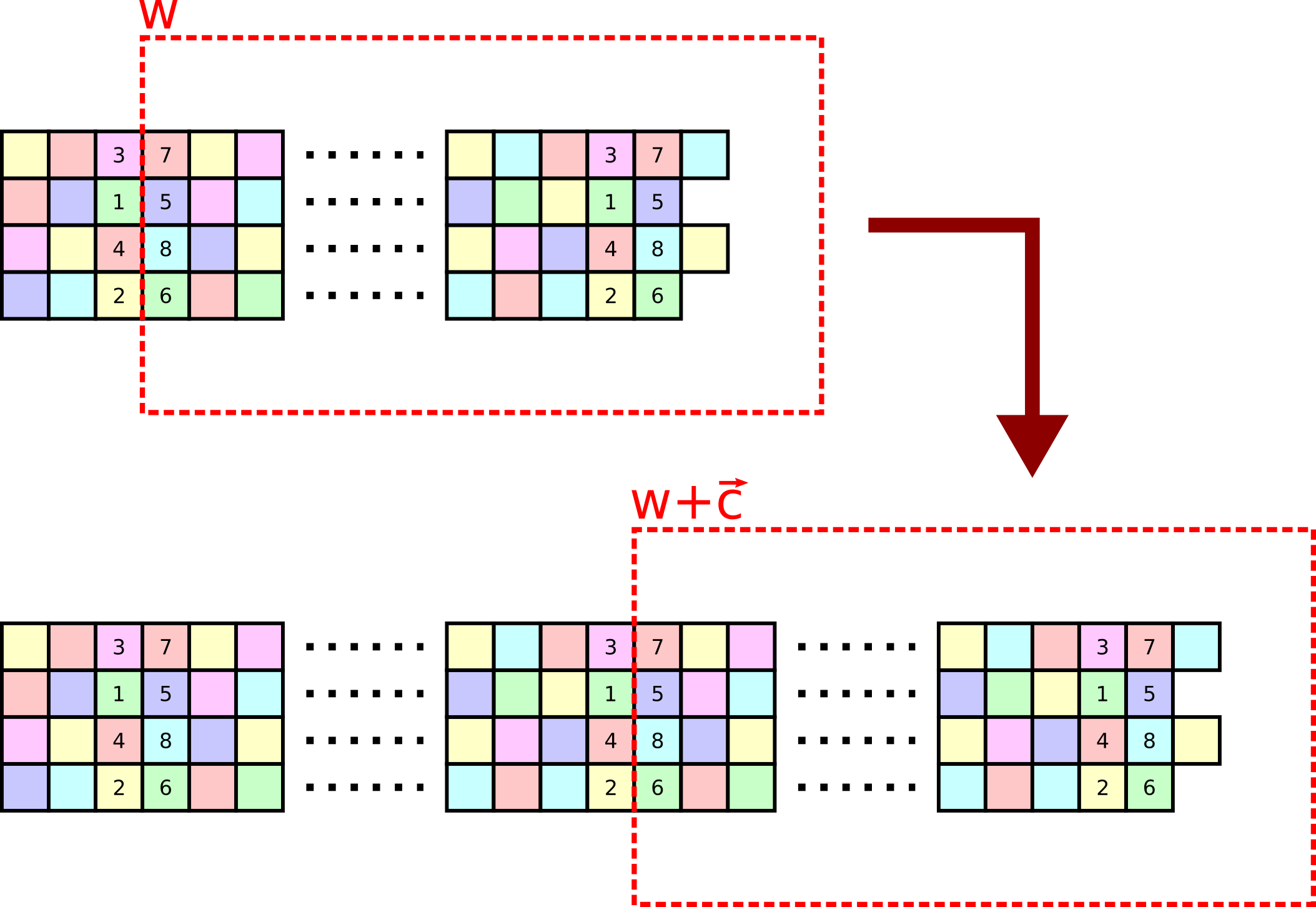}
    \caption{Using the Window Movie Lemma to ``pump'' assembly sequences. The top assembly depicts a ribbon of tiles growing horizontally to the right and numbers on tiles describe a relative order of attachment. If such a ribbon of tiles grows long enough, then by pigeonhole principle, eventually there must exist two identical vertical slices along its length. Because every tile attachment inside a window $w$ depends only on the tiles and their relative order of attachment along the window, we can thus find an assembly sequence where growth repeats after the second identical vertical slice. This can be performed indefinitely to ``pump'' the ribbon.
    }
    \label{fig:window-movie-lemma-pumping}
\end{figure}
\section{Results}\label{sec:results}

In this section we sketch our results. Full proofs can be found starting in Section~\ref{sec:dir-cant-sim-undir-tech}. We begin with some trivial observations which allow us to fill in several boxes from Table~\ref{tab:results}.

\begin{observation}\label{obs:subset-no-sim}
If there exists a directed system $\calT$ in tile-assembly model $M$ which cannot be simulated by any system in tile-assembly model $M'$, then (1) there exists a system in $M$ which cannot be simulated by any system in $M'$, (2) there exists a system in $M$ which cannot be simulated by any directed system in $M'$, and (3) there exists a directed system in $M$ which cannot be simulated by any directed system in $M'$.
\end{observation}

\begin{observation}\label{obs:2d-cant-sim-3d}
There exists systems, both directed and undirected, in the 3D models (3DaTAM and SaTAM) which cannot be simulated by any systems in any of the 2D models (aTAM and PaTAM, both directed and undirected).
\end{observation}

Observation~\ref{obs:subset-no-sim} holds because the set of directed systems in a model is a subset of all systems in that model. Consequently, $\calT$ is a system in both $M$ and in the directed subset of $M$. By assumption, $\calT$ cannot be simulated by any system in $M'$ and therefore cannot be simulated by any subset of systems of $M'$, particularly the subset of directed systems. Regarding Observation~\ref{obs:2d-cant-sim-3d}, while we restrict the notion of simulation to use square macrotiles, simulations of systems on triangular lattices have been implemented using roughly hexagonal macrotiles made from square tiles \cite{amoebots2019simulation}, so one might imagine the possibility that by loosening our definition of simulation to use more interesting macrotiles, it could be possible to capture the geometry of 3D square tiles using 2D tiles. In our case however, we note that there can exist no planar embedding of the lattice graph of $\mathbb{Z}^3$ as a consequence of Kuratowski's theorem. Consequently, there can be no way to divide $\mathbb{Z}^2$ into connected regions of macrotile locations which preserves the adjacency of points in $\mathbb{Z}^3$ and therefore simulation could not be possible even if we generalized our notion of macrotiles. This is true for any 3D systems which have producible assemblies whose domains, as graphs, are non-planar as is trivially possible in all 3D models considered.

\subsection{Simulations using existing tile sets}\label{sec:3d-tile-set-obs}

In \cite{DDDIU}, it was shown that there exists IU tile sets for the 3DaTAM, SaTAM, and both models' subsets of directed systems. While the main focus of that result was intrinsic simulation within a model, those IU tile sets can be used to trivially fill in a few boxes of Table~\ref{tab:results}. First we note that any aTAM system can also be thought of as a 3DaTAM system (or even SaTAM system since tiles occupying only a single plane of 3D space can't constrain a 3D region) with glues only appearing on 4 of the 6 faces of any tile. Second, we note that the IU tile sets for the 3DaTAM and SaTAM differed only by the addition of a few tile types responsible for growing a wall around each face of a macrotile before resolving. This was necessary for intrinsic universality in the SaTAM since without them, the tiles making up a macrotile were sparse enough to necessarily allow a diffusion path for tiles to pass through a resolved macrotile. Consequently, if we don't include those tile types, then the IU tile set can simulate 3DaTAM systems even in the SaTAM since without walls surrounding each macrotile, the diffusion restriction does not interfere with the attachment of any tiles. Finally, by design, this tile set preserves directedness when simulating a directed system. Therefore, using the IU tile set and proofs from \cite{DDDIU}, the following observations hold.

\begin{observation}\label{obs:3D-can-sim-2D-non-planar}
    There exists a universal tile set in both the 3DaTAM and SaTAM which intrinsically simulates all systems in the aTAM, preserving directedness.
\end{observation}

\begin{observation}\label{obs:SaTAM-can-sim-3DaTAM}
    There exists a universal tile set in the SaTAM which intrinsically simulates all systems in the 3DaTAM, preserving directedness.
\end{observation}

\subsection{Directed systems cannot simulate undirected systems}\label{sec:dir-cant-sim-undir}

\begin{restatable}{theorem}{DirCantSimUndir}
    \label{thm:dir-cant-sim-undir}
    There exist systems in the aTAM, 3DaTAM, PaTAM, and SaTAM, which cannot be simulated by any directed system in any of these models.
\end{restatable}

\begin{wrapfigure}[10]{r}{0.38\textwidth}
    \centering
    \begin{subfigure}[b]{0.4\linewidth}
        \centering
        \includegraphics[width=0.55\textwidth]{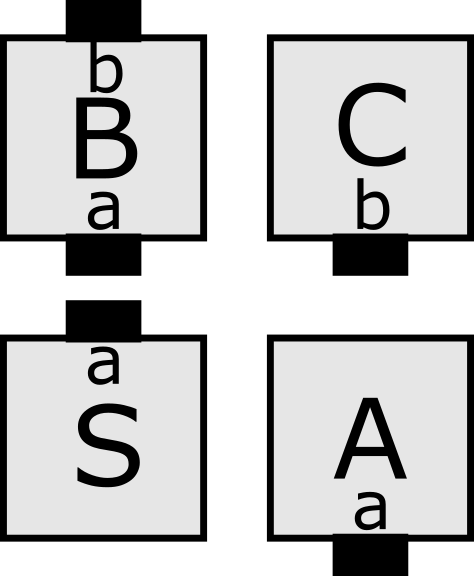}
        \caption{\label{fig:undirected-tile-set}}
    \end{subfigure}
    \hfill
    \begin{subfigure}[b]{0.48\linewidth}
        \centering
        \includegraphics[width=0.6\textwidth]{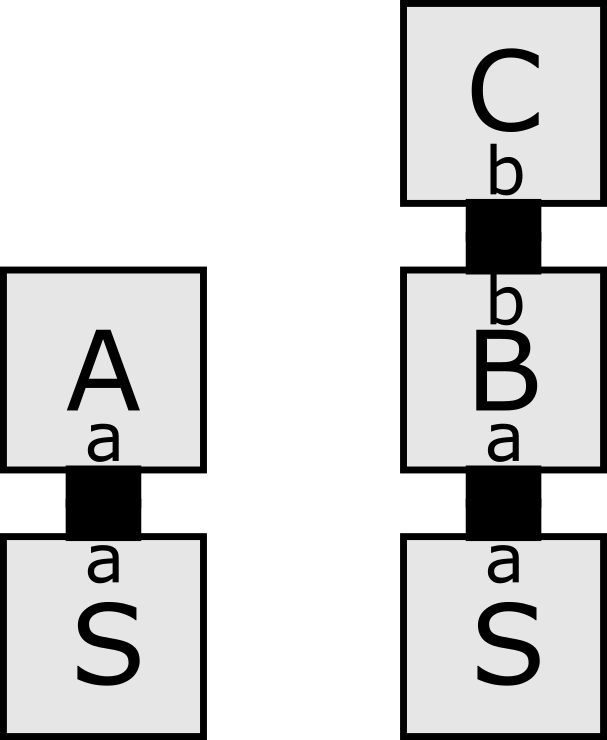}
        \caption{\label{fig:undirected-assemblies}}
    \end{subfigure}
    \caption{(a) Tile set of an undirected system for the proof of Theorem~\ref{thm:dir-cant-sim-undir} and (b) Its two terminal assemblies.}
    \label{fig:undirected-system}
\end{wrapfigure}

Whereas directed systems only have one terminal assembly, undirected systems can have several. Figure~\ref{fig:undirected-system} illustrates the tile set and terminal assemblies of a simple undirected system $\calT$ which can be a system in the aTAM, 3DaTAM, PaTAM, or SaTAM without modification as it does not use any dynamics unique to any of those models. Because directed systems can only have a single terminal assembly, any directed system attempting to simulate $\calT$ would necessarily fail since any assembly representation function $R^*$ could not map one terminal assembly to both terminal assemblies of $\calT$. A formal treatment of this proof is given in Section~\ref{sec:dir-cant-sim-undir-tech}, but the idea is the same.

\subsection{The PaTAM cannot simulate the aTAM}\label{sec:patam-cant-sim-atam}

Here we show that there are aTAM systems which cannot be correctly simulated by any PaTAM systems. To show this, we take advantage of the fact that aTAM systems are capable of growth inside of constrained regions while PaTAM systems are not. Specifically, we show that the PaTAM can't simulate the directed aTAM and, by Observation~\ref{obs:subset-no-sim}, note that this also implies that the PaTAM can't simulate the aTAM. Full technical details of the proof can be found in Section~\ref{sec:patam-cant-sim-atam-tech}.

\begin{restatable}{theorem}{PatamCantSimDirAtam}
    \label{thm:patam-cant-sim-dir-atam}
    There exists a system $\calT$ which is a directed aTAM system, and therefore also an aTAM system, which cannot be simulated by any PaTAM system. 
\end{restatable}


Figure~\ref{fig:blocking-counters} is a schematic diagram of the terminal assembly of $\calT$, a directed aTAM system which we claim is impossible to simulate in the PaTAM. Note that $\calT$ is more complex than a system in which tiles attach to constrain a region which could have another tile attach within. This is because the definition of intrinsic simulation allows for macrotiles to resolve even when they aren't completely filled with tiles. Consequently, while macrotiles may map to tiles constraining a region, the tiles making up the macrotiles may not constrain a region. Our construction is designed to
\begin{figure}
    \centering
    \includegraphics[width=0.6\linewidth]{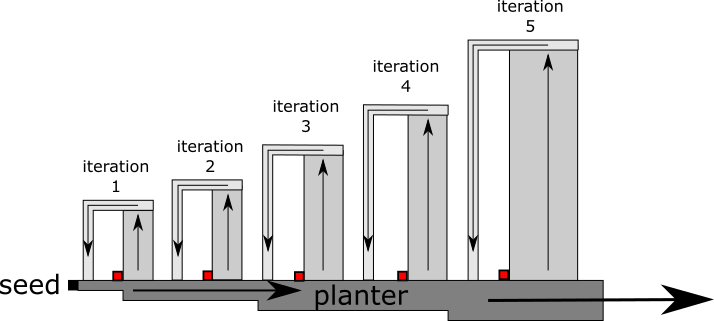}
    \caption{System $\calT$ of the proof of Theorem~\ref{thm:patam-cant-sim-dir-atam}. An infinite planter grows to the east from the seed and initiates upward growth of an infinite series of counters, each taller than the last, which initiate single-tile-wide paths that grow to the left then crash downward into the planter. To the left of each counter, at its base, it is possible for a red tile to attach.}
    \label{fig:blocking-counters}
\end{figure}
ensure that at some point, any supposed simulating system must constrain a region before the tiles inside are able to attach.
In our directed aTAM system $\calT$, this is done by first initiating the growth of a \emph{planter}, a gadget that counts up in binary as it grows eastward, initiating the growth of increasingly tall arms at defined intervals. These arms are essentially binary counter gadgets which each grow upward to a distance, encoded in the glues of the tiles provided by the planter, and initiate the growth of thin arms when they finish. The thin arms are just a single tile wide and begin by growing a fixed distance to the west before growing south to crash into the planter below. By this process, each arm initiated by the planter constrains increasingly large regions of space which each contain a single location between the planter and arms, in which a single tile can cooperatively attach (denoted by the red squares in Figure~\ref{fig:blocking-counters}). Each of the tiles making up the southward growing portion of the thin arms are of the same tile type, each with identical glues on their north and south faces. While it is possible for different macrotiles to map to the same tile in $\calT$, there are only so many combinations of tiles that make up a macrotile. Consequently, regardless of scale factor, if we look far enough down the planter, there will be an arm which grows tall enough that the simulating set must repeat a macrotile representation in two places along the same thin arm. We can then use the Window Movie Lemma to show that this arm ``pumps'' in our supposed simulating system, before crashing into the planter. It is therefore impossible for any simulating PaTAM system to prevent a region from becoming constrained before the macrotile inside is able to resolve, yielding terminal assemblies which aren't correctly mapped to a terminal assembly in $\calT$.

\subsection{The aTAM cannot simulate the PaTAM}\label{sec:atam-cant-sim-patam}

Given that the PaTAM is just the aTAM with an added restriction on tile attachment, it's not terribly surprising that the PaTAM can't simulate the full dynamics of the aTAM; however, less obvious is the fact that the planarity restriction \emph{also} gives the PaTAM some capabilities not possible in the aTAM, namely the ability to constrain a region and stop growth within. We utilize this ability in our proof of Theorem~\ref{thm:atam-cant-sim-patam} which is sketched here and detailed in Section~\ref{sec:atam-cant-sim-patam-tech}. Also, by Observation~\ref{obs:subset-no-sim}, this also holds for the directed aTAM.


\begin{restatable}{theorem}{AtamCantSimPaTAM}
    \label{thm:atam-cant-sim-patam}
    There exists a PaTAM system $\mathcal{P}$ which cannot be simulated by any aTAM system.
\end{restatable}

As with the proof for Theorem~\ref{thm:patam-cant-sim-dir-atam}, in the definition of intrinsic simulation, we consider all possible representation functions and scale factors to prove impossibility. Figure~\ref{fig:atam-cant-sim-patam} is a schematic diagram of PaTAM system $\calP$ which is impossible to correctly simulate in the aTAM. 
\begin{wrapfigure}{r}{0.3\textwidth}
    \centering
    \includegraphics[width=0.8\linewidth]{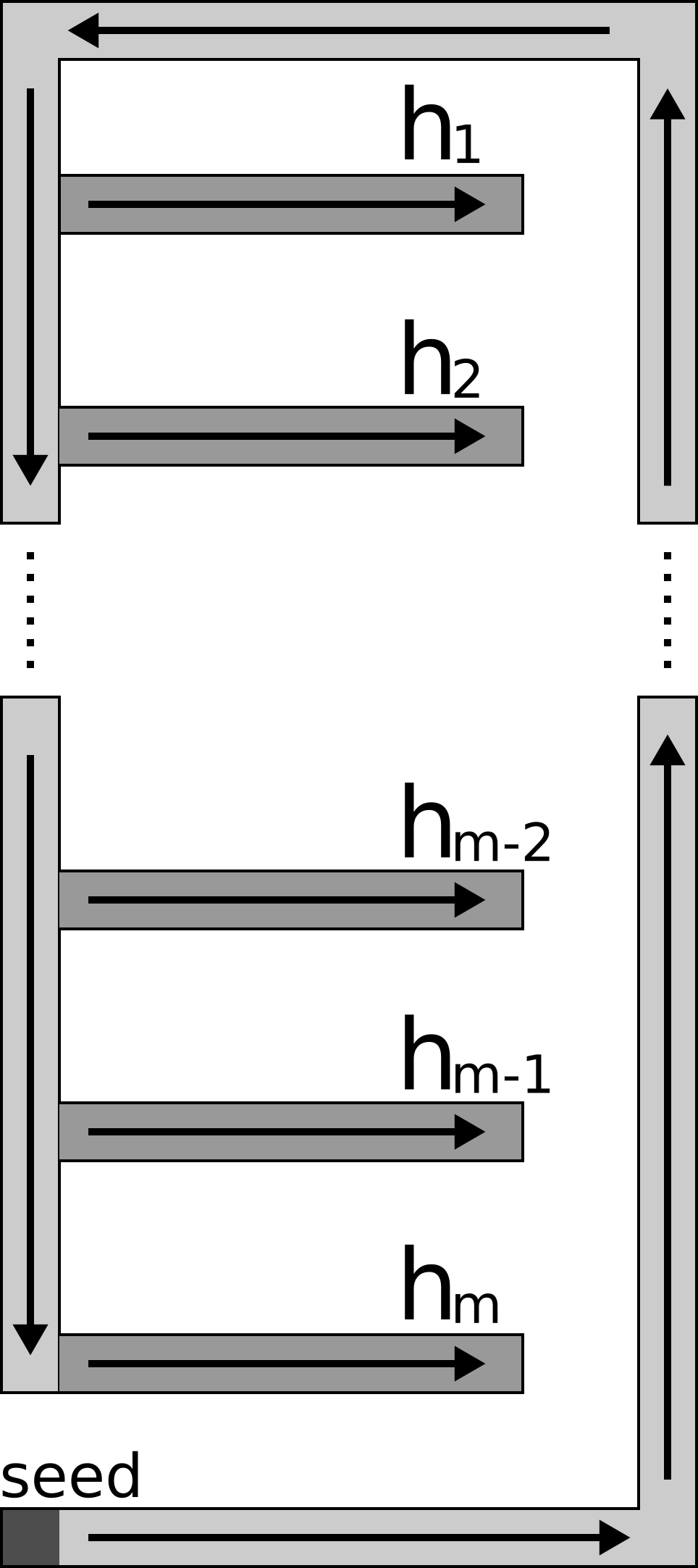}
    \caption{A schematic of system $\mathcal{P}$ for the proof of Theorem~\ref{thm:atam-cant-sim-patam}. Tiles grow in a rectangular shape, periodically spawning arms which can crash into the walls and constrain a region. It is undirected and its size depends non-deterministically on the number of tiles that attach between each corner.}
    \label{fig:atam-cant-sim-patam}
\end{wrapfigure}
Growth of $\calP$ begins with tiles attaching in a row growing east. The length of this row is non-deterministic as at any point along the row, it's possible for a corner tile to attach, initiating growth to the north. Consequently, $\calP$ is an undirected system so any potential simulating system must be able to simulate all possible assemblies of $\calP$. Similarly, northward and eastward growing rows of tiles attach with some length depending on how many tiles attached before each corner. Finally, a column of tiles begins growing south and, as it does, initiates the growth of several arms eastward, each spaced 4 tiles apart. Both the southward growing column of tiles and the arms continue growth until they are constrained or crash into another part of the assembly.
To show that $\calP$ cannot be simulated in the aTAM, we assume the existence of a simulating aTAM system $\calT$ and prove that it must admit some assembly sequences which don't correspond to those in $\calP$. To do this, we consider an assembly sequence in $\calP$ where the rectangle of tiles grows to a size, based on the scale factor of the simulation, so that a sufficiently large number of sufficiently long arms are spawned by the south growing column of tiles. We also choose an assembly sequence where the south growing column will eventually collide with the seed tile, constraining the region containing the arms. Because we've chosen the assembly to be sufficiently large, each arm is capable of being ``pumped'' as per the window movie lemma. We then grow the bottom arm until just after it has collided with the east wall and note that, while $\calT$ is an aTAM system and can still grow tiles inside of the constrained region, tiles on the inside and outside will no longer be able to affect each other's growth. There are a few cases to be considered, depending on whether or not the representation function has resolved the last tile of the bottom arm, but essentially we then show that we can continue the growth of the west wall until its macrotiles have resolved to tiles in $\calP$ that constrain the rectangle's interior. By a counting argument and our choice of the number of arms, we can then show that one of the other arms must be able to continue growth within the constrained region, and that the assembly sequence in $\calT$ maps to one invalid in $\calP$.

\subsection{The 3DaTAM cannot simulate the SaTAM}

The proof of Theorem~\ref{thm:3datam-cant-sim-satam} is similar in principle to the proof of Theorem~\ref{thm:atam-cant-sim-patam}, albeit with a slightly different system which takes advantage of the differences between 2D and 3D. We sketch the idea here and formally prove it in Section~\ref{sec:3datam-cant-sim-patam-tech}.

\begin{restatable}{theorem}{aTAMThreeCantSimSaTAM}
    \label{thm:3datam-cant-sim-satam}
    There exists an SaTAM system $\mathcal{S}$ which cannot be simulated by any 3DaTAM system.
\end{restatable}

The system $\calS$ for this result, as illustrated in Figure~\ref{fig:3datam-cant-sim-satam}, initially grows 2 nearly sealed chambers connected by a thin tunnel which allows for a diffusion path between them. These chambers both have a fixed base size of $9\times 9$, but they can grow to have an arbitrary height in a way similar to the frame of the system used in the proof of Theorem~\ref{thm:atam-cant-sim-patam}. Once fully grown, the ceiling of one chamber contains a single tile wide opening which is the only way for tiles to diffuse into the chambers from outside; we call the chamber with this hole the \emph{outer chamber} and the other one the \emph{inner chamber}. Additionally, from the bottoms of both chambers, pillars can grow upwards to an arbitrary height by the attachment of copies of tiles with identical tile types. The pillar in the inner chamber will eventually crash into the ceiling \emph{or} until the pillar in the outer chamber grows tall enough to plug the opening in its ceiling and constrain the space inside.
\begin{wrapfigure}[16]{l}{0.4\textwidth}
    \centering
    \includegraphics[width=0.9\linewidth]{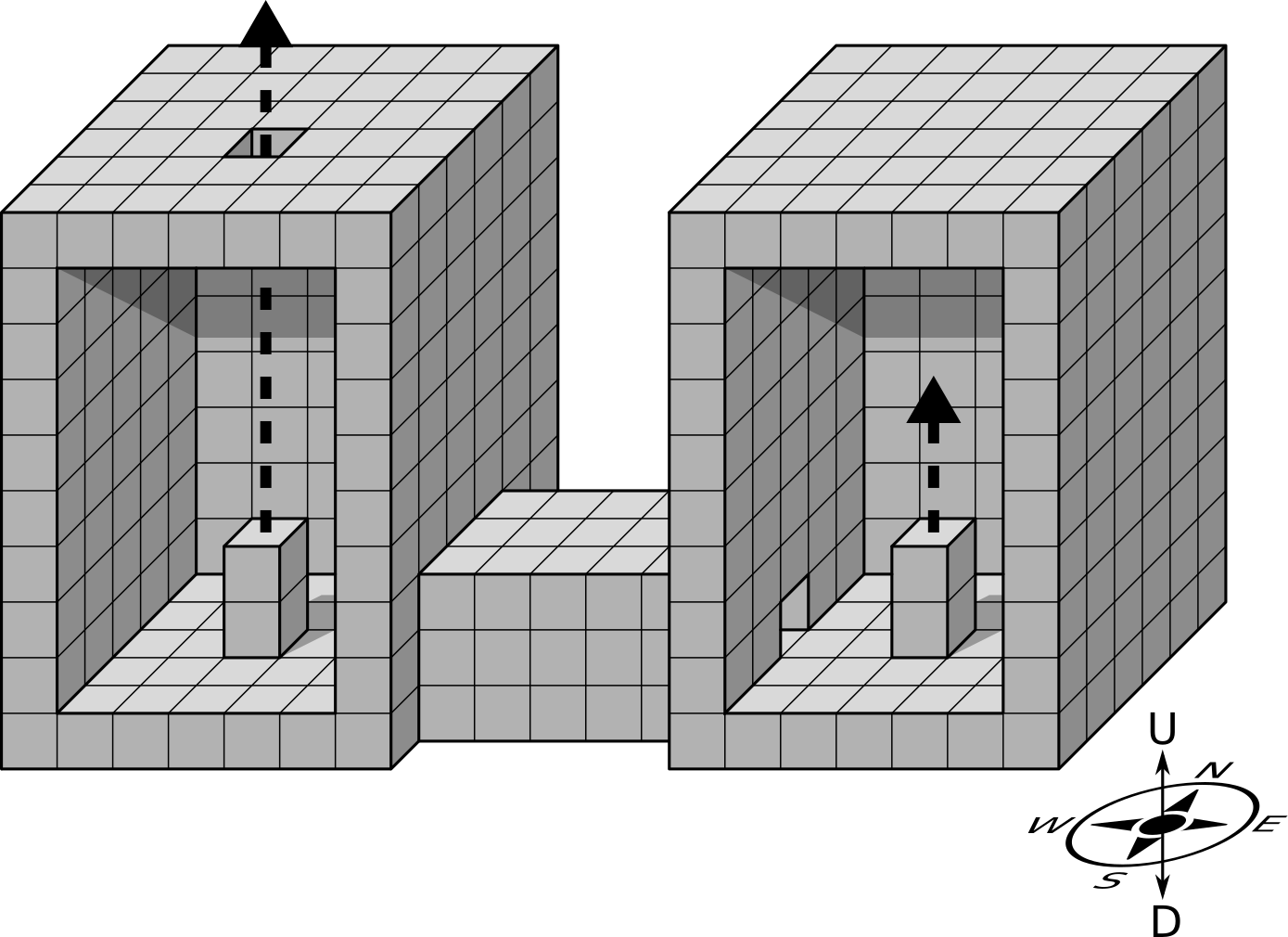}
    \caption{Cut-away view of system $\calS$ from the proof of Theorem~\ref{thm:3datam-cant-sim-satam}. Two chambers are connected by a thin tunnel. Pillars growing inside the outer west chamber will eventually constrain the region within the chambers, at which point, the pillar growing in the inner east chamber will no longer be able to continue growth.}
    \label{fig:3datam-cant-sim-satam}
\end{wrapfigure}
We show that $\calS$ cannot be simulated by any 3DaTAM system by showing that, in any potential simulating system, under the right conditions, although unwanted, it must still possible for the inner chamber pillar to continue growth even after the outer chamber pillar has sealed the chambers. To do this, we note that during some supposed simulation, the only way for the pillar in the inner chamber to ``know'' that the chambers have been sealed, is for tiles to attach inside of the tunnel. Consequently, because the tunnel is thin with a cross-section made of a hollow $3\times 3$ square, the chambers can only communicate with each other a finite amount of times during a simulation. Specifically, if the scale factor of the simulation is $c$, then the number of tiles that can be placed in any $x$-coordinate corresponding to the tunnel is bounded by $5c \times 5c$ which includes any potential tiles growing in the fuzz adjacent to the macrotiles of the tunnel. Therefore, by a simple counting argument, if we initially grew our chambers to have a sufficiently large height, then there must exist some assembly sequence where both pillars grow by any desired number of macrotiles (which we choose to be long enough to allow pumpable growth) and during which no tile is placed in the center of the tunnel. Using the Window Movie Lemma, we then construct an assembly sequence where the outer chamber pumps to constrain the chambers. Because during this assembly sequence, no tiles are placed in the center of the tunnel, there is nothing to stop the inner chamber pillar from also being pumped. Such an assembly sequence must be possible in any 3DaTAM system which supposedly simulates our system $\calS$, and since this assembly sequence corresponds to one which is invalid in the SaTAM, such a simulation is impossible.

\subsection{The PaTAM can simulate the directed PaTAM}\label{sec:patam-can-sim-dpatam}

\begin{restatable}{theorem}{PaTAMCanSimDPaTAM}
    \label{thm:patam-can-sim-dpatam}
    There exists a universal Planar aTAM tile set $S$ that can simulate any directed PaTAM system.
\end{restatable}

Despite the fact that both the PaTAM and directed PaTAM are not intrinsically universal for themselves\cite{DDDIU}, using tools from \cite{DDDIU} and \cite{IUSA} we are able to construct a PaTAM tile set capable of simulating arbitrary directed PaTAM systems. Here we outline the process by which a PaTAM tileset $S$ can simulate any given directed PaTAM system $\calT$. The tileset $S$ is universal, meaning that regardless of the directed PaTAM system $\calT$, the same tileset will be used at a fixed binding threshold, with only the seed of the simulating system changing to accommodate $\calT$. Technical details can be found in Section~\ref{sec:patam-can-sim-dpatam-tech}. 

\begin{wrapfigure}{r}{0.45\textwidth}
    \centering
    \includegraphics[width=\linewidth]{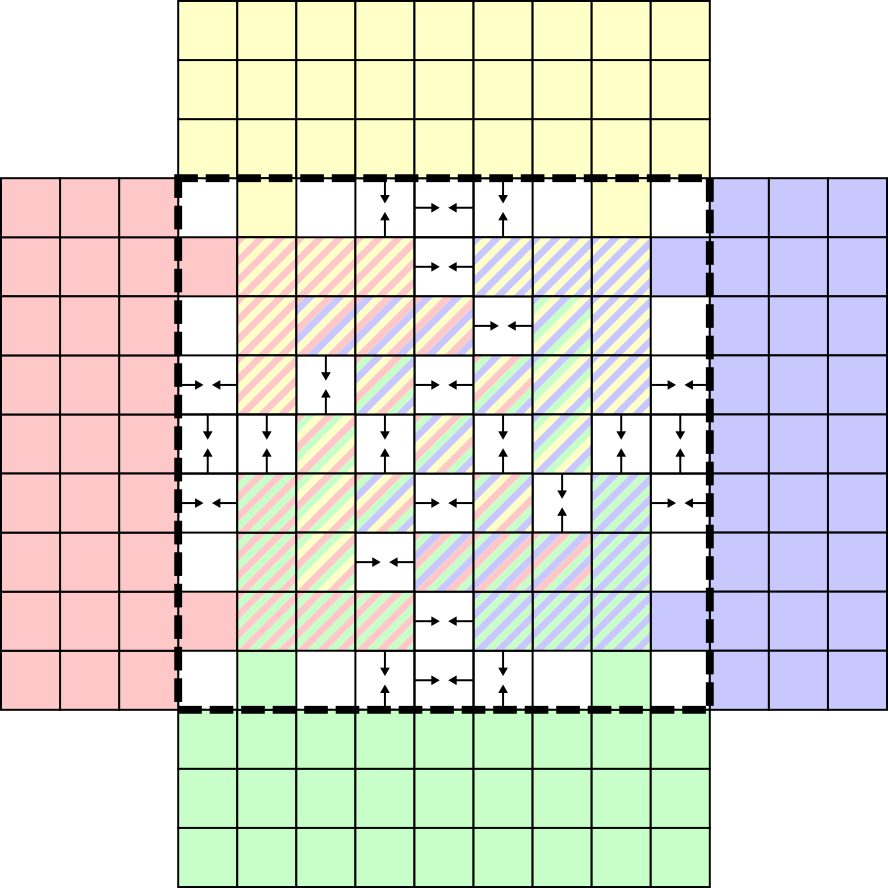}
    \caption{A schematic describing the $9\times 9$ grid of potential component blocks which may appear in a macrotile location. Squares containing two arrows indicate a grid location which may contain a probe region. The surrounding macrotiles are illustrated using colored tiles to represent their relative direction from the current macrotile. Colors of CB locations indicate which surrounding macrotile the CB may have information about.}
    \label{fig:sim-macrotile-schematic-body}
\end{wrapfigure}
Given a directed PaTAM system $\calT$, we define a simulating system $\calS$ using a fixed tile set at binding threshold $2$. The seed of $\calS$ consists of already-resolved macrotiles in the same configuration as the seed of $\calT$. Each macrotile in $\calS$ consists of a $9\times 9$ grid of structures we call \emph{component blocks} (CBs) which are each made of many smaller tile-based constructions and which each store an encoding of the system $\calT$ along with a bit of extra data in the form of specific glues on some of its tiles. The CBs of a macrotile each perform calculations using tiles which emulate Turing machines to determine how they should grow and whether or not the macrotile can resolve given the current information regarding the surrounding macrotiles.

Each CB essentially behaves like an individual tile on the $9\times 9$ grid and we can think of CBs as growing in one of two ways. Either the CB grows using tile attachments from another adjacent CB in a way analogous to a $\tau$-strength tile attachment, or a CB can grow in the gap between two adjacent CBs in certain locations of the grid designated as \emph{probe regions}. This is analogous to a tile attachment that occurs by cooperative binding between two opposing tiles (which we refer to as \emph{across-the-gap} cooperation). These ``cooperative attachments'' between CBs are used to consolidate information between the CBs.
For instance, one CB might contain information encoded about the north adjacent macrotile and one might contain information about the west; in the probe region between them, a new CB can grow which will contain the information about both which it can then use to determine if a tile attachment in $\calT$ would be possible in the tile location corresponding to the macrotile. Figure~\ref{fig:sim-macrotile-schematic-body} illustrates the layout of a macrotile into CB locations with these probe regions indicated by squares with two opposing arrows.

\begin{wrapfigure}[11]{l}{0.28\textwidth}
    \centering
    \includegraphics[width=0.9\linewidth]{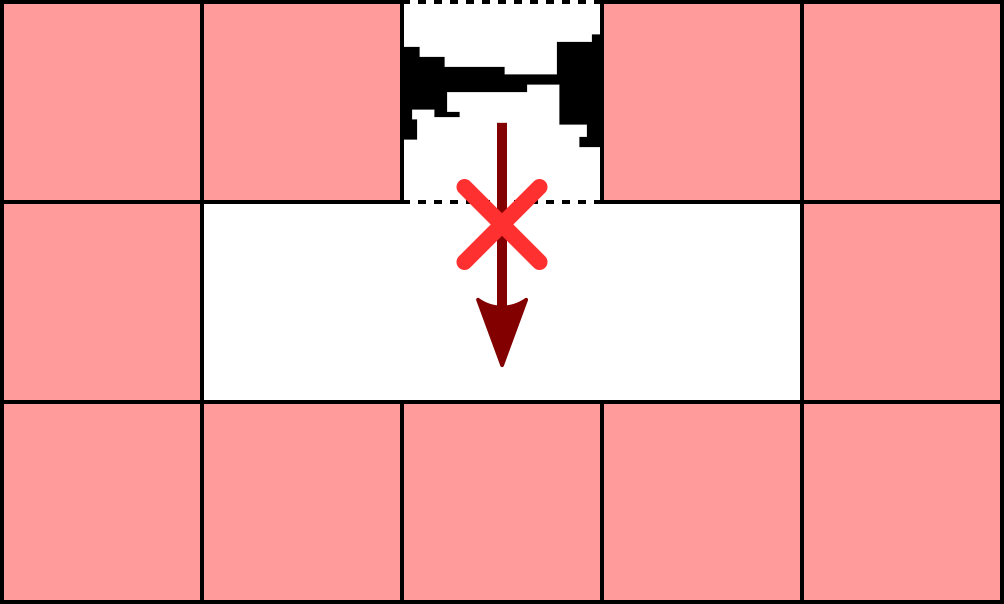}
    \caption{When checking for across-the-gap cooperation during a simulation, tiles can't naively span the entire gap without disconnecting two regions of space.}
    \label{fig:sim-atg-body}
\end{wrapfigure}
Probe regions are CB locations in which two adjacent CBs, on opposite sides, can present structures called \emph{probes} which are long, thin structures that grow from the surrounding CBs towards the center of a CB location. Each probe that grows in a probe region, indicates some possible combination of information from surrounding macrotiles and grows in a unique position according to this information. The length of a probe is chosen to be just shy of the center of the CB location, so that when two probes align from opposing sides of the probe region, there will be exactly a single tile wide gap between them. This gap allows a tile to cooperatively attach and grow along the sides of the probes to recover the information from both. Otherwise, if no probes in a probe region align, there will be enough room for the components that make up a CB to squeeze in between the probes from one side of the probe region to another. Figure~\ref{fig:component-block-probes-body} illustrates two scenarios involving probe regions.

\begin{wrapfigure}[18]{r}{0.5\textwidth}
    \centering
    \includegraphics[width=0.9\linewidth]{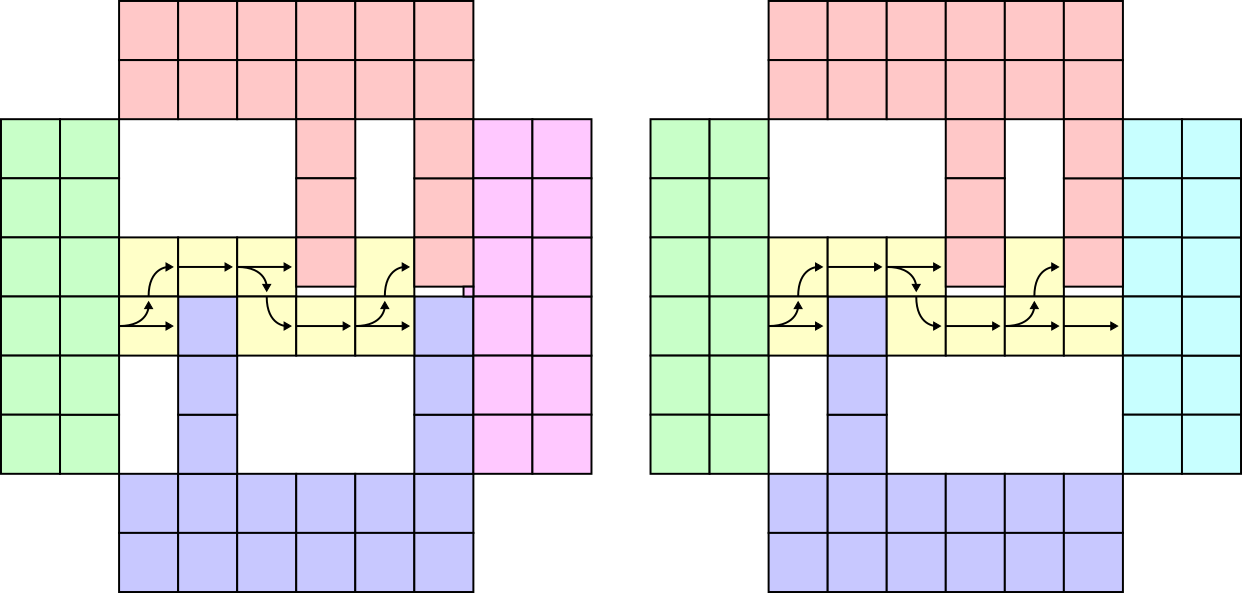}
    \caption{Probe regions between component blocks. The red and blue CBs grow probes to the center of the CB location in the middle while the green CB attempts to grow through the probe region. On the left, two probes happen to align, in which case a path of tiles containing information from the green CB cannot pass and the CB to the east results from the cooperative tile attachment between the probes. On the right, no probes align meaning the path of tiles from the green CB can squeeze between the probes to influence the growth of the CB to the east.}
    \label{fig:component-block-probes-body}
\end{wrapfigure}

Probe regions were introduced in \cite{IUSA} to solve the problem illustrated in Figure~\ref{fig:sim-atg-body}. Naively when simulating a tile system, to check for macrotiles which may cooperate across-the-gap, tiles must grow to query both adjacent macrotiles and determine if the attachment is possible. This however necessarily separates regions of space and in the case of planar systems also constrains one before it has been determined if the attachment can even occur. 
If it cannot, then tiles will no longer be able to attach in the constrained region and the simulation will likely end up being invalid. Probes avoid this problem by aligning exactly when across-the-gap cooperation is possible while still allowing tile structures to grow through if they don't align.

\begin{figure}[ht!]
    \centering
    \includegraphics[width=0.6\textwidth]{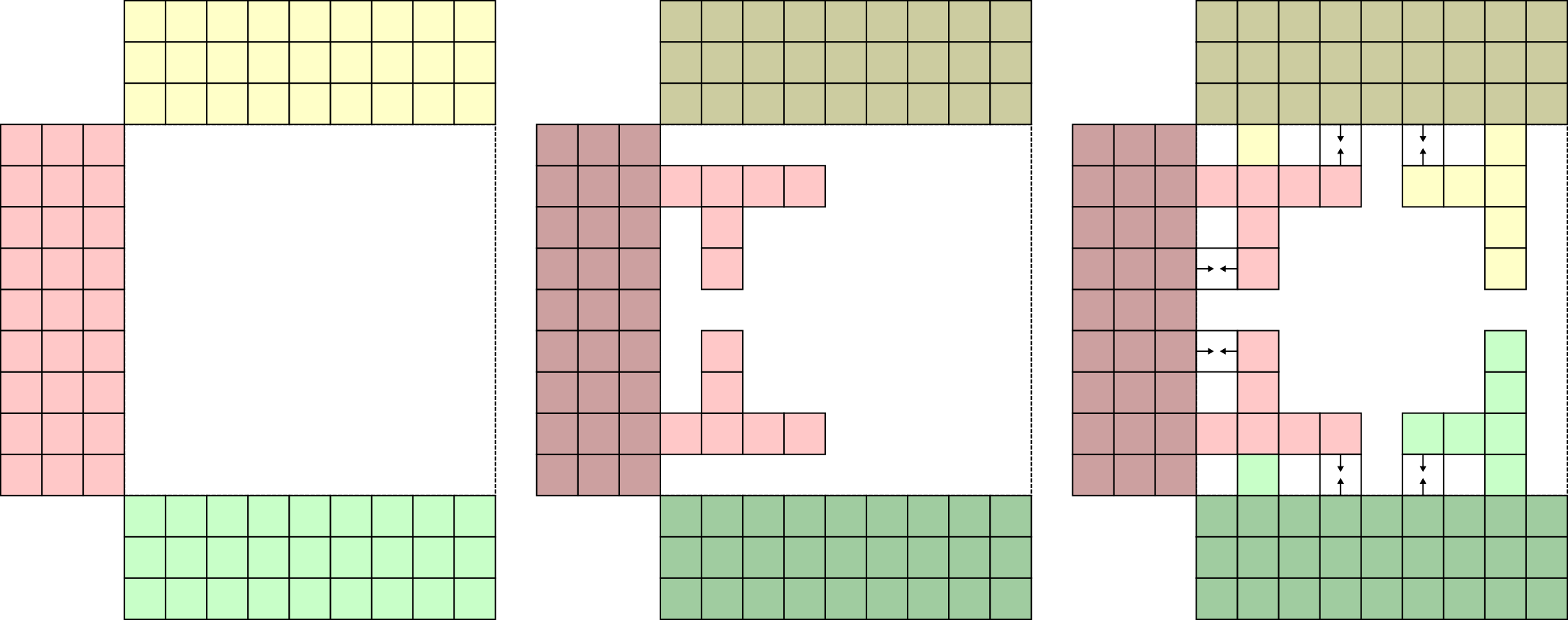}
    \caption{Hands made of component blocks growing from surrounding macrotiles.}
    \label{fig:sim-hands}
\end{figure}

Now that we have an idea of how the component blocks and probe regions behave we describe the protocol for resolving a macrotile by highlighting a few important cases. Growth within a macrotile begins when one or more of the surrounding macrotiles resolve and tiles begin to attach within the macrotile. From a surrounding macrotile, the protocol always begins by the growth of two ``T''-shaped structures made from CBs called \emph{hands} illustrated in Figure~\ref{fig:sim-hands}. Note that two adjacent surrounding macrotiles may both attempt to grow hands in the same location. This is handled by a single point of competition and the first surrounding macrotile for placing a tile in the closest corner of the shared hand locations is allowed to place theirs. Between the hands and the surrounding macrotiles probes are grown in the regions indicated on the right of Figure~\ref{fig:sim-hands} which allows a CB to ``attach'' cooperatively to combine information from both the hand and nearby macrotile. In some cases this information may be redundant, but with two or more surrounding macrotiles at least one location will always be able to combine information from two macrotiles. 

\begin{figure}[ht!]
    \centering
    \includegraphics[width=0.9\textwidth]{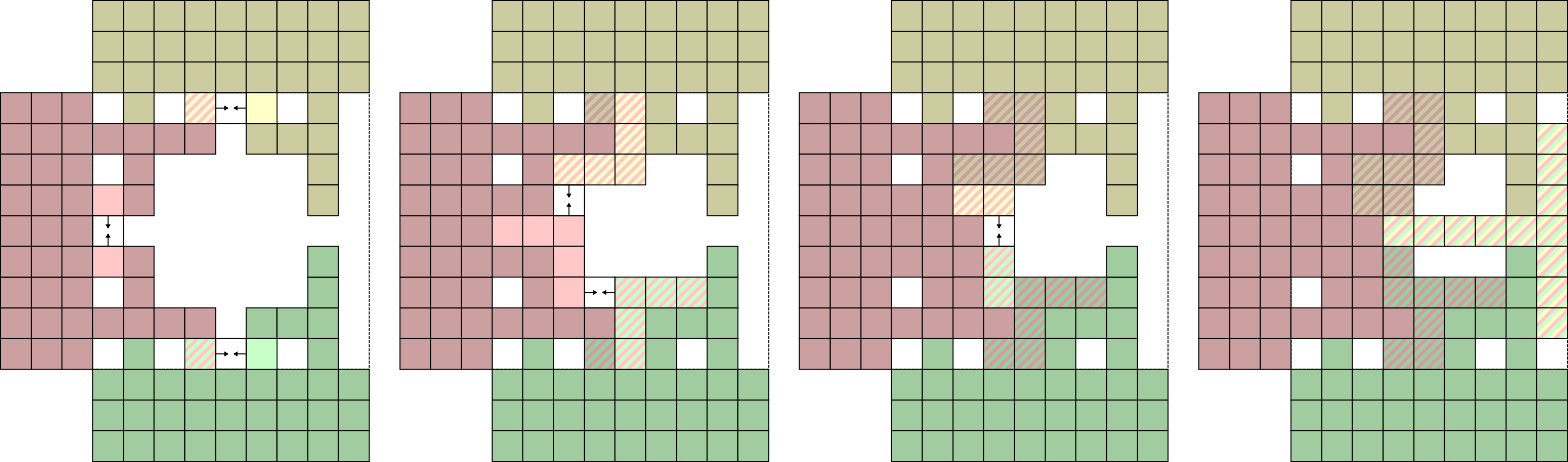}
    \caption{Once the hands have grown, CBs cooperate until information from all sides has been combined into a single CB. Then the macrotile can resolve.}
    \label{fig:sim-3-surrounding}
\end{figure}

The CBs resulting from cooperation between the hands and surrounding macrotiles then cooperate once again and CBs grow along the hands to form clockwise elbows with additional probe regions between them. CBs then cooperatively attach between these elbows and cooperate again near the center of the tile to eventually combine all of the information from the surrounding macrotiles. Once this occurs, the CB which ``attaches'' in the center of the tile contains the information from all sides. If the surrounding macrotiles represent tiles in $\calT$ capable of placing a tile, additional CBs can grow to the remaining sides to present this information to the remaining sides and repeat the procedure in the adjacent macrotile locations.

\begin{figure}[ht!]
    \centering
    \includegraphics[width=0.9\textwidth]{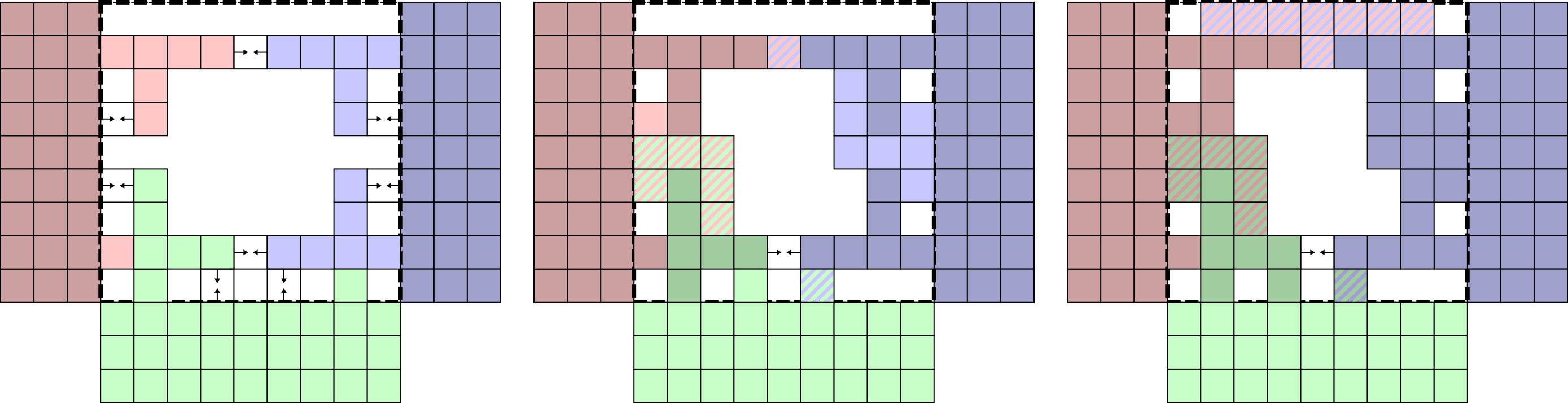}
    \caption{Probe regions between opposing macrotiles can check for across-the-gap cooperation.}
    \label{fig:sim-atg-coop}
\end{figure}

In the case that an across-the-gap cooperation is possible in $\calT$, the protocol deviates slightly. Illustrated in Figure~\ref{fig:sim-atg-coop}, if across-the-gap cooperation is possible between the east and west macrotiles, their hands will share a probe region with aligned probes. Consequently, a CB can grow in that location and resolve the macrotile. This growth may constrain the region to the south, halting any tile attachments and CB growth in the south side of the macrotile, but this doesn't matter since the macrotile will only need to start the process in adjacent macrotiles that haven't yet resolved. The described protocol is robust to different orders of hand growth and different numbers of surrounding macrotiles, including those that don't end up contributing to macrotile resolution. If at any time a CB has sufficient information to determine how the macrotile should resolve, it begins growth to the center and then surrounding edges of the macrotile. This process will not be interrupted by other CBs since we are simulating a directed system where at most one unique tile can attach in each location.

\section{Technical Appendix: Omitted Details and Proofs}\label{sec:proofs}
\subsection{Directed systems cannot simulate undirected systems}\label{sec:dir-cant-sim-undir-tech}

\DirCantSimUndir*

\begin{proof}
Let $T$ be the tile set shown in Figure~\ref{fig:undirected-tile-set}. Define a tile assembly system $\mathcal{T} = (T,\sigma,1)$ where $\sigma$ consists of a single tile: a copy of the tile labeled $S$ located at the origin. Note that if $\mathcal{T}$ is a system in the aTAM, 3DaTAM, PaTAM, or SaTAM, it behaves identically: it is an undirected system with exactly two terminal assemblies, which are shown in Figure~\ref{fig:undirected-assemblies}.

We prove Theorem~\ref{thm:dir-cant-sim-undir} by contradiction. Therefore, assume that $\mathcal{S} = (S,\sigma', \tau)$ is a directed system which simulates $\calT$ (under representation function $R$ and at scale factor $c$). By the definition of simulation, $R(\sigma') = \sigma$ (i.e. the seed $\sigma'$ represents $\sigma$), and there exists at least one assembly sequence in $\mathcal{S}$ such that in the resulting terminal assembly the $c\times c$ macrotile region north of $\sigma'$ represents, under $R$, the tile of $T$ labeled $A$, and there also exists at least one assembly sequence in $\mathcal{S}$ such that in the resulting terminal assembly the $c\times c$ macrotile region north of $\sigma'$ represents, under $R$, the tile of $T$ labeled $B$. However, since $\mathcal{S}$ is directed, there can only be one terminal assembly, and therefore both assembly sequences result in the exact same tiles being placed in that $c\times c$ macrotile region. But, since $R$ is a function, it cannot map the same input macrotile region to two different tiles. This is a contradiction, and therefore $\mathcal{S}$ does not simulate $\calT$, and since the only assumption made about $\mathcal{S}$ was that it was directed, no directed system can simulate $\calT$ (whether it is in the aTAM, 3DaTAM, PaTAM, or SaTAM) and Theorem~\ref{thm:dir-cant-sim-undir} is proven.

\end{proof}
\subsection{The PaTAM cannot simulate the aTAM}\label{sec:patam-cant-sim-atam-tech}

\PatamCantSimDirAtam*




\begin{proof}
Let $\calT$ be the system which is schematically depicted in Figure~\ref{fig:blocking-counters}. $\calT$ has a seed consisting of a single tile placed at $(0,0)$ and has a binding threshold of 2 (i.e. $\tau=2$). From the seed tile, $\calT$ grows a ``planter'', which is simply a modified log-width binary counter which counts from 8 to $\infty$. It is modified so that as it counts each number $8 \le n < \infty$, the bits of that $n$ are rotated upward so that they can initiate a counter which grows upward $n$ rows. Additionally, there are $8$ extra spaces between the bits of each counter. The planter grows infinitely to the right, independently of each upward growing counter which it initiates. Each upward growing counter, seeded with the bits of a unique value of $n$, grows upward $n$ rows. At that point, it grows a single additional row across the top and then a single-tile-wide arm $4$ tiles to the left, which then allows a tile to attach to its south which has $\tau$-strength glues allowing copies of itself to attach to its north and south. This results in a single-tile-wide column composed of copies of that tile which grows downward until it eventually ``crashes'' into the planter (i.e. a tile of the column is placed adjacent to the planter so that no additional tiles can be placed). Note that the growth of the planter is designed so that the location of the planter into which a downward growing arm crashes must be completed before the counter which eventually initiates the growth of that arm can begin. Finally, at any time after the growth of the first row of the upward growing counter a red tile can bind to the leftmost tile of that row.

By careful design of the modules of $\calT$ (which are standard modules in many aTAM constructions, see \cite{jCCSA,DirectedNotIU}, e.g.), it is clear that $\calT$ is not only a valid aTAM system, but it is also directed. We prove Theorem~\ref{thm:patam-cant-sim-dir-atam} by contradiction, so therefore assume that $\mathcal{P}$ is a PaTAM system with tileset $P$ which simulates $\calT$. Let $c$ be the scale factor by which $\mathcal{P}$ simulates $\calT$, and let $R$ be the representation function.

For each $8 \le n < \infty$, we use the term ``$n$th \emph{iteration}'' to refer to the growth of the upward counter to $n$, the arm which grows to the left then downward, and the red tile associated with that $n$. Notice that the downward growing arms become arbitrarily tall, but remain a constant width. By the definition of simulation and inspection of $\calT$, we see that number of tiles spanning any row of an arm cannot be more than $3c$, which is the width of $3$ macrotiles. This number includes the macrotile representing the tile of the arm, plus one macrotile of fuzz on each side. Anything outside that width would violate the constraints on fuzz in the simulation and make the simulation invalid. Therefore, let $p=(3c)!(|P|+1)^{3c}$ and note that this is an upper bound on the number of orders in which tiles from $|P|$ (including the lack of a tile) can be placed in a row of $3c$ tile locations.

Now, consider some assembly sequence $\vec{\alpha}$ which grows the assembly up to the $2p$th iteration where the counter grows to a height of $2p$, but where the $2p$th red tile has not yet attached. Since by assumption $\calP$ simulates $\calT$, there must be a corresponding assembly sequence $\vec{\alpha}'$ in $\calT$. Notice that by our definition of $p$ and the size of the iteration, tiles must attach in the same way and order on the top rows of at least two distinct macrotiles (and surrounding fuzz) corresponding to the downward growing arm. We can then define two windows $w_1$ and $w_2$, both as the boundary of a $3c \times p+1$ rectangle with the tops centered along the top rows of these macrotiles and note that during $\vec{\alpha}'$ both windows will have the same window movies. Consequently we can construct a new assembly sequence $\vec{\beta}'$ in $\calT$, similar to $\vec{\alpha}'$ except that tiles attach identically in the regions enclosed by both windows. These tile attachments can then be repeated until blocked by some tile in the planter macrotiles or surrounding fuzz. At this point, since $\mathcal{P}$ is a Planar aTAM system, it is impossible for tiles to attach in the macrotile region representing the red tile for that iteration. Regardless of how tiles attach after this, the macrotile corresponding to the red tile will never be able to resolve and thus the simulation will be incorrect. Since, $\mathcal{P}$ fails to simulate $\calT$, and since we made no assumptions about $\mathcal{P}$ other than the fact that it is a PaTAM system which simulates $\calT$, no such simulator can exist.
\end{proof}


\subsection{The aTAM Can't Simulate the PaTAM}\label{sec:atam-cant-sim-patam-tech}

\AtamCantSimPaTAM*

\begin{figure}[h!]
    \centering
    \includegraphics[width=0.7\textwidth]{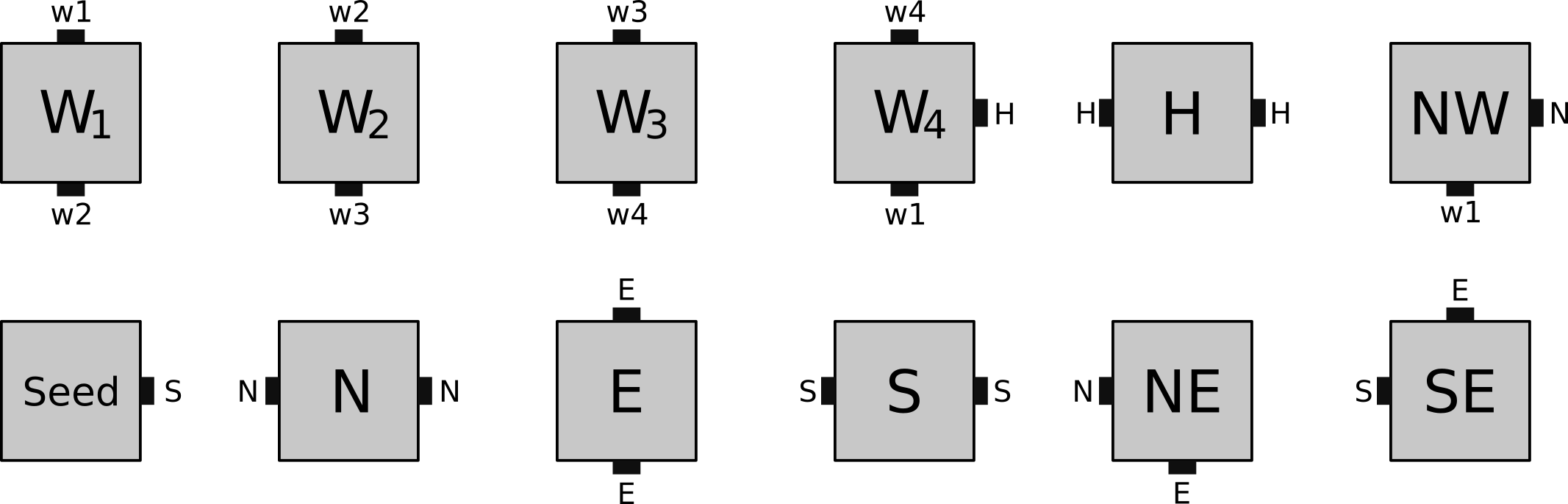}
    \caption{Tileset for a PaTAM system which cannot be correctly simulated by any aTAM system.}
    \label{fig:aTAM-cant-sim-PaTAM-tileset}
\end{figure}

\begin{proof}
    
    To prove theorem~\ref{thm:atam-cant-sim-patam}, let $\mathcal{P}$ be the PaTAM system illustrated in Figure \ref{fig:aTAM-cant-sim-PaTAM-example-system}. This system starts with a single seed tile from which a rectangular frame grows. The south, east, and north walls of this frame are each made of several copies of a single tile type unique to that side, while the west wall is made of 4 distinct tile types which grow in a periodic sequence. Because the frame's first 3 sides each consist of copies of a single tile type, the length of the each side is non-deterministic and depends on how many copies of each type happen attach before a corner tile. Additionally, there is no corner tile which attaches after the west wall tiles so it will grow indefinitely or until it collides with another tile. As the west wall grows, each 4th tile initiates the eastward growth of an arm which grows by the attachment of identical tiles indefinitely or until the region is constrained or a collision occurs. All glues in this system are strength-1 and the binding threshold is likewise 1.
    
    We show that no system in the aTAM is capable of simulating all assembly sequences of $\mathcal{P}$ by contradiction. Therefore, assume that $\mathcal{T} = (T,\sigma,\tau)$ is an aTAM system which simulates $\mathcal{P}$, let $c$ be the scale factor, and let $R$ be the representation function of the simulation. Now, let $p = (3c)!\cdot(2(g+1))^{3c}$. This is an upper bound on the number of orders in which tiles attachments from $T$ (including the lack of a tile) be placed in along the boundary between 2 rows of $3c$ tile locations. Consequently, $p$ bounds the number of possible 1D slices of scale-$c$ macrotiles, including two adjacent fuzz macrotile regions, accounting for the relative order in which the tiles attach. This can be thought of as a pumping length of sorts since, if a line of at least $p$ identical tiles is growing in $\mathcal{P}$, sufficiently far enough away from other tiles, then it must be the case that, between at least two of the corresponding macrotiles in $\mathcal{T}$, an identical sequence of tile attachments occurred. When this happens, it's then possible to consider an assembly sequence wherein any tile attachments beyond the second of these identical macrotile boundaries mimics the attachments beyond the first. This implies the existence of assembly sequences in $\mathcal{T}$ with periodic growth which can continue indefinitely or until blocked by other parts of the assembly.
    
    \begin{figure}
        \centering
        \includegraphics[width=0.4\textwidth]{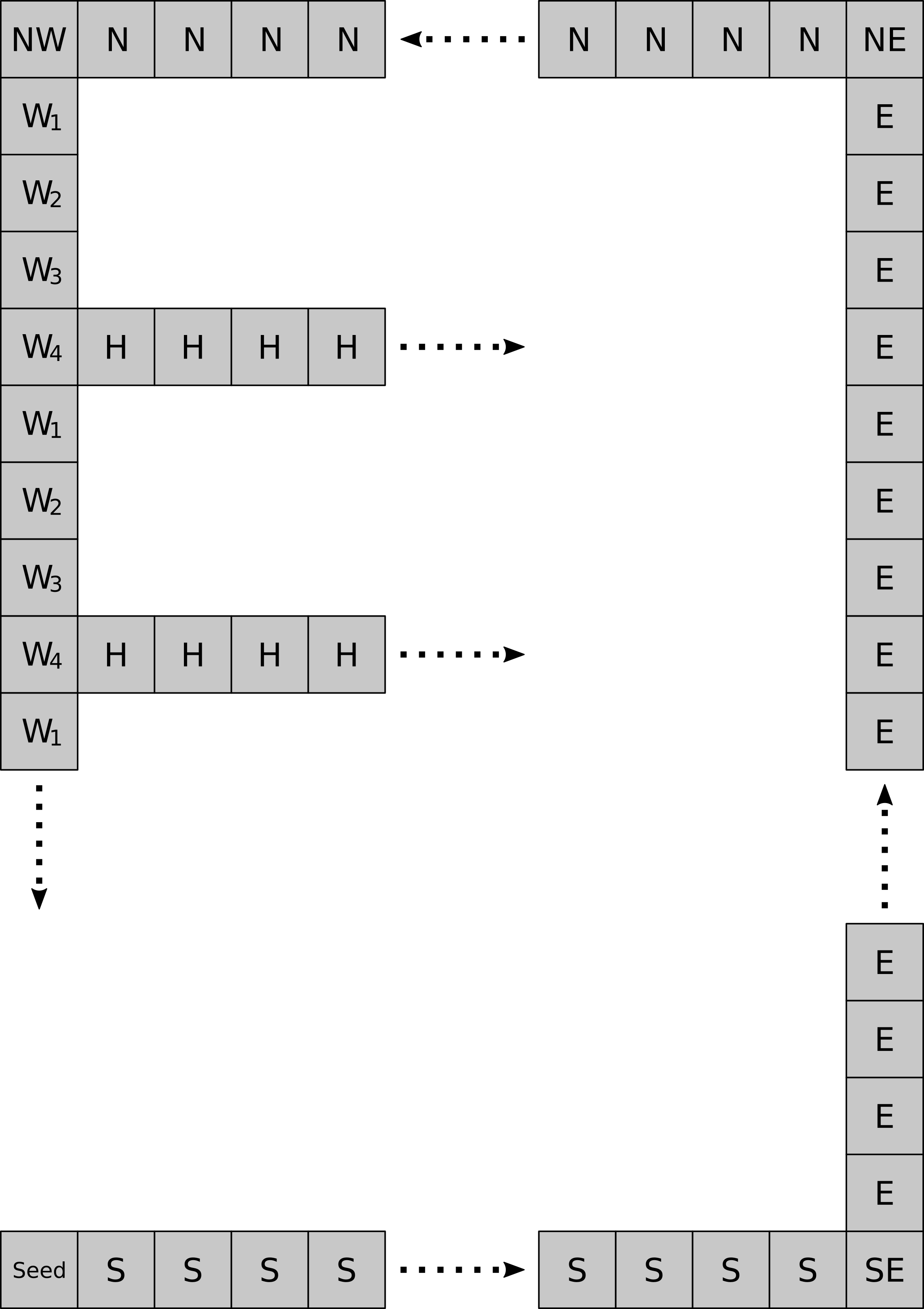}
        \caption{This PaTAM system cannot be correctly simulated by any aTAM system.}
        \label{fig:aTAM-cant-sim-PaTAM-example-system}
    \end{figure}
    
    Consider now an assembly $\alpha_0$ in $\mathcal{P}$ wherein the backward \emph{C} shaped frame grows so that its north and south sides are of length $p(6c+2)$ and its east side is of length $4\cdot(6c + 2) + 2$ not including the corner tiles. The reason for these specific values will be explained shortly. Additionally, in $\alpha_0$, the west side has grown to the point where it is two tiles away from colliding with the seed tile, but has not yet initiated the growth of any of it's horizontal arms. Let $\vec{\alpha}_0$ be the assembly sequence in $\mathcal{P}$ which starts with the seed tile and ends with $\alpha_0$. So far, it should not be difficult to see how system $\mathcal{T}$ could simulate this assembly sequence. Let $\vec{\alpha}_0'$ be some assembly sequence in $\mathcal{T}$ which models $\vec{\alpha}_0$.
    
    Next, we will consider a sequence of assembly sequences $\vec{\alpha}_1, \vec{\alpha}_2, \ldots, \vec{\alpha}_n$ in $\mathcal{P}$, each of which follows from the previous, wherein horizontal arms grow from the west side in a specific manner. We chose the value $n=6c + 1$ for reasons which will become clear soon. By our choice of the east side length, the west side will be able to spawn $m=6c + 2$ horizontal arms (each spaced 4 tiles apart) while still remaining 2 tiles away from colliding with the seed tile. We will call these arms $h_1, h_2, h_3, \ldots, h_m$ from north to south for convenience. Additionally, let $\tilde{y}$ be the y-coordinate which is exactly between the y-coordinates of arms $h_{m-1}$ and $h_m$ in $\mathcal{P}$. Accordingly, fix $\tilde{y}'$ to be any y-coordinate within the row of macrotile locations in $\mathcal{T}$ corresponding to the tile locations in $\mathcal{P}$ at y-coordinate $\tilde{y}$.
    
    We define the assembly sequence $\vec{\alpha}_i$ ($i=1,2,\ldots,n$) immediately following the growth of assembly sequence $\vec{\alpha}_{i-1}$ such that: (1) tiles attach so as to grow arm $h_i$ until it collides with the east side, and (2) $p$ tiles attach to the end of arm $h_m$. Note that in each assembly sequence arm $h_m$ only grows partially by $p$ tiles. We've chosen the width of our assembly so that even if there is an assembly sequence corresponding to each arm north of $h_m$, it will not collide with the east side. Now for each assembly sequence $\vec{\alpha}_i$ in $\mathcal{P}$, let $\vec{\alpha}'_i$ be the corresponding assembly sequence in $\mathcal{T}$. Corresponding to each tile attachment in $\vec{\alpha}_i$, there may be several tile attachments in $\vec{\alpha}'_i$. During these attachments in the macrotile locations of $\mathcal{T}$, we will keep track of a specific condition, namely tiles being placed at y-coordinate $\tilde{y}'$. Including the fuzz regions surrounding the east and west arms, there are only $6c$ tile locations in $\mathcal{T}$ at y-coordinate $\tilde{y}'$ where tiles could be placed which are not too far away to invalidate the simulation. Because of our choice of $n=6c + 1$, at least one of the assembly sequences, say $\vec{\alpha}'_j$, must occur without a tile being placed at y-coordinate $\tilde{y}'$. 
    
    Now we define the assembly sequence $\vec{\beta}$ in $\mathcal{P}$ as follows. First we follow the assembly sequences $\vec{\alpha}_0,\vec{\alpha}_1,\ldots,\vec{\alpha}_{j-1}$ in order. Then we follow assembly sequence $\vec{\alpha}_j$ up to the attachment of the second to last tile of arm $h_j$. We then deviate from our assembly sequences and skip the attachment of the last tile of arm $h_j$. Instead, we grow arm $h_m$, stopping one tile short of colliding with the east side of our assembly. Finally, we attach the remaining 2 tiles of the west side, colliding with the seed tile. It shouldn't be too difficult to see that $\vec{\beta}$ is a valid assembly sequence in $\mathcal{P}$. Additionally, notice that $\vec{\beta}$ is a terminal assembly in $\mathcal{P}$, since by the planarity constraint, it's now impossible for any of the arms inside the assembly to continue growth. Therefore, if we define $\vec{\gamma}$ to be the assembly sequence $\vec{\beta}$ followed by the attachment of the remaining tile in arm $h_j$, $\vec{\gamma}$ would not be a valid assembly sequence in $\mathcal{P}$. Despite this, we can construct an assembly sequence $\vec{\gamma}'$ in $\mathcal{T}$ which models $\vec{\gamma}$, proving that $\mathcal{T}$ does not correctly simulate $\mathcal{P}$.
    
    To do this, we construct the assembly sequence $\vec{\gamma}'$ in $\mathcal{T}$ as follows. First, we follow the assembly sequences $\vec{\alpha}'_0,\ldots,\vec{\alpha}'_{j-1}$. Then, we follow assembly sequence $\vec{\alpha}'_j$ up to but not including the tile attachment which would cause the macrotile corresponding to the final tile of $h_j$ to resolve. Next we continue following $\vec{\alpha}'_j$, but we skip any tile attachments north of y-coordinate $\tilde{y}'$. Because during $\vec{\alpha}'_j$ no tiles attached in this y-coordinate, this does not interfere with the growth of the macrotiles corresponding to the next $p$ tiles in arm $h_m$. During these attachments south of $\tilde{y}'$, we consider a rectangular window $w$ with dimension $p+1 \times 3c$ and note that by our choice of $p$, it must be possible to position translate $w$ along the arm in two ways, both with identical window movies. By the Window Movie Lemma, all tile attachments possible in the first of these translated windows must be possible in the second. These tile attachments can then be repeated (or ``pumped'') until blocked by a tile on the east side. Note that because the first pumped sequence of tile attachments didn't grow north of $\tilde{y}'$, the remaining attachments won't either. If, upon collision with the east side, the macrotile corresponding to the final tile of $h_m$ resolves, then we reach a contradiction since we can then resolve the final macrotile corresponding to $h_j$ with a single tile attachment after the arm $h_m$ supposedly blocked off the region containing $h_j$. Otherwise, we can continue with an arbitrary sequence of tile attachments in $\mathcal{T}$ corresponding to the remaining two tiles of the west side in $\mathcal{P}$. These tile attachments cannot influence the region closed off by the collision between the macrotiles corresponding to $h_m$ and the east wall, and therefore the attachment of a single tile in the macrotile corresponding to the final tile of $h_j$ will still lead to the macrotile resolving, a contradiction.
\end{proof}

\subsection{The 3DaTAM cannot simulate the Spatial aTAM}\label{sec:3datam-cant-sim-patam-tech}

Here we use the convention that the cardinal directions north, south, east, west, up, and down correspond to $+y$, $-y$, $+x$, $-x$, $+z$, and $-z$ respectively. When referring to dimensions of tile constructions, we use width, length, and height to refer to dimensions in the $x$, $y$, and $z$ axis respectively.

\aTAMThreeCantSimSaTAM*

\begin{figure}
    \centering
    \includegraphics[width=0.6\textwidth]{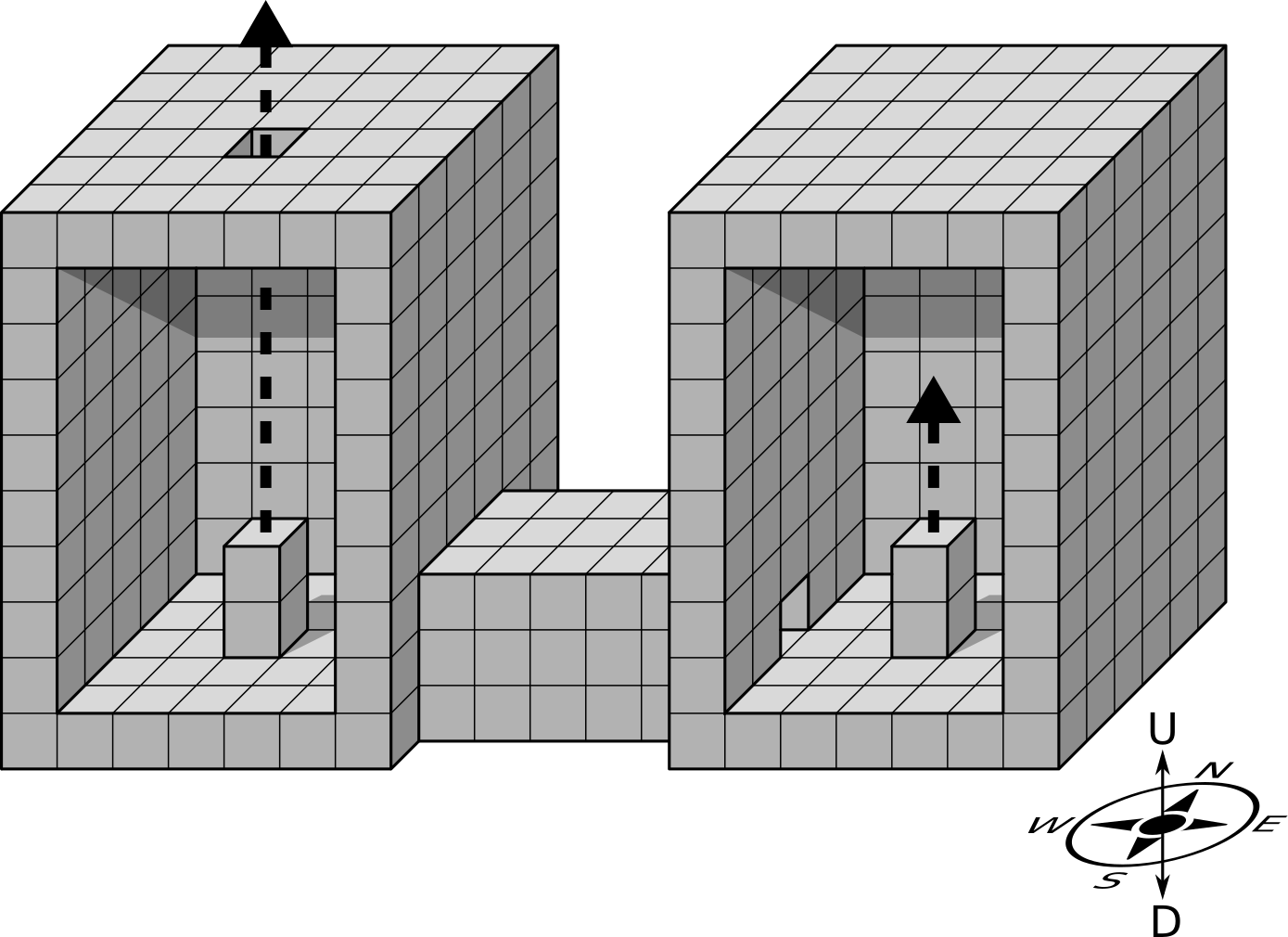}
    \caption{A cut-away view of the system $\calS$ used in the proof of Theorem~\ref{thm:3datam-cant-sim-satam}. Two chambers are connected by a thin tunnel.}
    \label{fig:3datam-cant-sim-satam-tech}
\end{figure}

\begin{proof}

Let $\calS$ be the SaTAM system, illustrated in Figure~\ref{fig:3datam-cant-sim-satam-tech}, described as follows. From the seed, tiles attach to form two boxes, one on the east which we will call the \emph{inner chamber} and one on the west which we will call the \emph{outer chamber}, connected by a thin tunnel which separates the chambers by a distance of $5$. The base of each chamber is a $9 \times 9$ square of tiles and each chamber can grow to have an arbitrary height before a special tile attaches to initiate the growth of their ceilings. The ceiling of the inner chamber is solid, but the ceiling of the outer chamber has a single tile opening in its center. Inside each chamber, a pillar of tiles can grow upwards from the center of the base. These pillars are each made of copies of the same tile type which can attach on top of each other allowing the pillars to grow arbitrarily tall. The tunnel has a cross section of a hollow $3 \times 3$ square allowing for tiles to diffuse into the inner chamber until the pillar in the outer chamber has grown tall enough to plug the opening in the ceiling and constrain the region inside.

Now suppose, for contradiction, that there exists a 3DaTAM system $\calT$ which simulates $\calS$ using tileset $T$, scale factor $c$, and macrotile representation function $R$. First, we define a few constants whose values will be important during the proof. Let $p=(9c^2)!(|T|+1)^{9c^2}$. This is an upper bound on the number of orders in which tiles from $T$ (including no tile) can be placed in a $3c\times 3c$ square. We also define $b=25c^2$ and let $h=(p + 1)(b + 2) + 2$ which will be used as the height of our chambers and whose value will become clear shortly. We, now consider a few assembly sequences in $\calS$ which, by our assumption, $\calT$ must be able to simulate. First let $\vec{\alpha}^\calS_0$ be the assembly sequence in $\calS$ during which tiles attach to the seed to complete the growth of both chambers so that both have an interior height of $h$ (i.e. not including the base and ceiling) and both pillars grow to a height of 2. We'll refer to the last assembly in $\vec{\alpha}^\calS_0$ as $\alpha^\calS_0$. Additionally, we consider a series of continuations of this assembly sequence which we will define inductively; so for $k=0,\ldots,b$, let $\vec{\alpha}^\calS_{k+1}$ be the assembly sequence which begins at the assembly $\alpha^\calS_k$ and during which the outer chamber pillar grows by $p+1$ tiles followed by the inner chamber pillar growing by $p+1$ tiles. To complete the inductive definition, let $\alpha^\calS_{k+1}$ be the final assembly in the assembly sequence $\vec{\alpha}^\calS_{k+1}$. Notice that, during the assembly sequences $\alpha^\calS_k$ for $k=1,\ldots,b+1$, each pillar will grow by a height of $p+1$ tiles which, by our definition of $h$ to be $(p+1)(b+2) + 2$, means that neither pillar has reached the ceiling yet. We then define $\vec{\alpha}^\calS$ to be the concatenation of each of these assembly sequences in order.

Because we assumed that $\calT$ simulates $\calS$, there must be an assembly sequence in $\calT$ which simulates the growth of assembly sequence $\vec{\alpha}^\calS$. Let $\vec{\alpha}^\calT$ be such assembly sequence and, for $k=0,\ldots,b+1$, let $\alpha^\calT_k$ be the first assembly in $\vec{\alpha}^\calT$ which maps under $R^*$ to $\alpha^\calS_k$. Then for convenience, we divide the assembly sequence $\vec{\alpha}^\calT$ into subsequences $\vec{\alpha}^\calT_k$ ($k=0,\ldots,5c+1$) so that $\vec{\alpha}^\calT_k$ simulates the assembly sequence $\vec{\alpha}^\calS_k$ and ends with the assembly $\alpha^\calT_k$. We will now use these assembly sequences in $\calT$ to construct a new assembly sequence in $\calT$ which cannot possibly correspond to a valid assembly sequence in $\calS$. First, let $x_t$ be the $x$-coordinate of the center of the tunnel in $calT$. For the macrotiles in the inner chamber to ``know'' anything about the macro tiles in the outer chamber, tiles must attach in a location with $x$-coordinate $x_t$. In $\calT$, a cross section of the tunnel at $x$-coordinate $x_t$ is a $3\times 3$ macrotile square and, if we include fuzz, this means that no tile in $\calT$ will be able to attach outside of the $5c \times 5c$ square surrounding the tunnel at $x$-coordinate $x_t$. Consequently, at most $25c^2$ (our choice of value for $b$) tiles will be able to attach in locations with $x$-coordinate $x_t$ during the assembly sequence $\vec{\alpha}^\calT$. Therefore, there must exist some index $j$ ($1\le j \le b+1$), such that during the assembly sequence $\vec{\alpha}^\calT_j$, no tile is ever placed at $x$-coordinate $x_t$.

Now notice that during assembly sequence $\vec{\alpha}^\calT_j$, tiles attach to grow both pillars by a height of $p+1$. Our definition of $p$ will allow us to use the Window Movie Lemma as follows. First, let $w$ be the window defined as the boundary of a box with $x$ and $y$ dimensions $3c$ and with $z$ dimension $p+2$. Then note that because each pillar started at a height of 2 macrotiles before any of the assembly sequences $\vec{\alpha}^\calT_1,\ldots,\vec{\alpha}^\calT_{b+1}$, none of the fuzz adjacent to any macrotiles in $\vec{\alpha}^\calT_j$ will be adjacent to any macrotiles except those which resolve in $\vec{\alpha}^\calT_j$. Because $p$ was defined as an upper bound on the number of orders in which tiles from $T$ can attach in a $3c \times 3c$ square and, since we are growing $p+1$ macrotiles on each pillar during $\vec{\alpha}^\calT_j$, there must be two ways to translate our window $w$, say $w^\text{out}_1$ and $w^\text{out}_2$, so that it is centered on the outer chamber pillar, such that $w^\text{out}_1$ and $w^\text{out}_2$ have identical window movies during the assembly sequence $\vec{\alpha}^\calT_j$. By the Window Movie Lemma, we can therefore create a new assembly sequence $\vec{\beta}^\text{out}$ which begins the same as $\vec{\alpha}^\calT_j$ but during which all of the tile placements in the region bounded by $w^\text{out}_1$ also occur in the region bounded by $w^\text{out}_2$. We can do the same for the pillar in the inner chamber to define an assembly sequence $\vec{\beta}^\text{in}$ analogously. We can then repeat this process to ``pump'' both pillars upwards by repeating the tile attachments occurring in the respective window regions, noting again that no tile is ever placed in $x$-coordinate $x_t$ so neither pillar's growth can affect the other.

We can therefore construct an assembly sequence $\vec{\gamma}$ in the following way. First we continue to pump the growth of the outer chamber pillar until the final macrotile resolves so that the pillar has grown to height $h$. We do have to be a bit careful though, because during the final pumping iteration, the assumptions of the Window Movie Lemma will no longer hold. This is because when we get close enough to the opening in the ceiling, the macrotiles of the ceiling and surrounding fuzz will alter the window movie. If we carefully consider the growth that occurs during the pumping however, this is not a problem. To see why, recall that growth is not allowed to occur in any locations which are not in the fuzz adjacent to resolved macrotiles. Consequently, in the prior pumping iterations, as tiles attach in a macrotile location, say $m$, on top of the pillar, no tiles have been able to attach in those macrotile locations adjacent to $m$ before it resolves. Consequently, even though the Window Movie Lemma no longer holds, we can still repeat the same sequence of tile attachments up to the point where the macrotile resolves to form a height $h$ pillar, as none of the tiles in the ceiling fuzz could prevent the necessary attachments. After growing the outer chamber pillar to a height of $h$ macrotiles, we continue the assembly sequence $\vec{\gamma}$ by performing the same tile attachments that occurred in $\vec{\beta}^\text{in}$, corresponding to tiles attaching to the inner chamber pillar. We then let the assembly sequence $\vec{\delta}$ be the assembly sequence formed by concatenating the assembly sequences $\vec{\alpha}^\calT_0,\ldots,\vec{\alpha}_{j-1}$ and the assembly sequence $\vec{\gamma}$. This assembly sequence is valid in $\calT$, but under $R$ it maps to an assembly sequence in $\calS$ wherein the inner chamber pillar continues growth after the outer chamber pillar has grown tall enough to constrain the region. This is a contradiction and since we made no prior assumptions about $\calT$, it cannot be the case that $\calT$ correctly simulates $\calS$.
\end{proof}

\subsection{The PaTAM Can Simulate the Directed PaTAM}\label{sec:patam-can-sim-dpatam-tech}

\PaTAMCanSimDPaTAM*

\subsubsection{Tile Gadgets}

This construction makes use of a few \emph{tile gadgets}, components which perform some specific task or grow in a specific pattern, which have been adapted from the existing tile-assembly literature. Much like how algorithms are often described using pseudocode using existing data-structures, we will describe our construction in terms of these gadgets.

\paragraph*{Turing Machine Blocks}

Used in~\cite{SolWin07} to assemble finite shapes with asymptotically optimal tile-complexity, a Turing machine block is a tile gadget which simulates space and time bounded Turing machines within a square region of space. Each side of a TM block is designated either as an input or an output. Both the input and outputs of a TM block are encoded by the glues presented along a row of tiles. These glues can be used to encode data in the form of bits if special glues are designated to act as 0s and 1s, and additionally, the glues presented by the output sides can be used to interface with other tile gadgets or initiate the growth of an adjacent TM block. Consequently, TM blocks can grow to form arbitrary computable patterns.

\paragraph*{Probes and Probe Regions}

\begin{figure}
    \centering
    \includegraphics[width=0.4\linewidth]{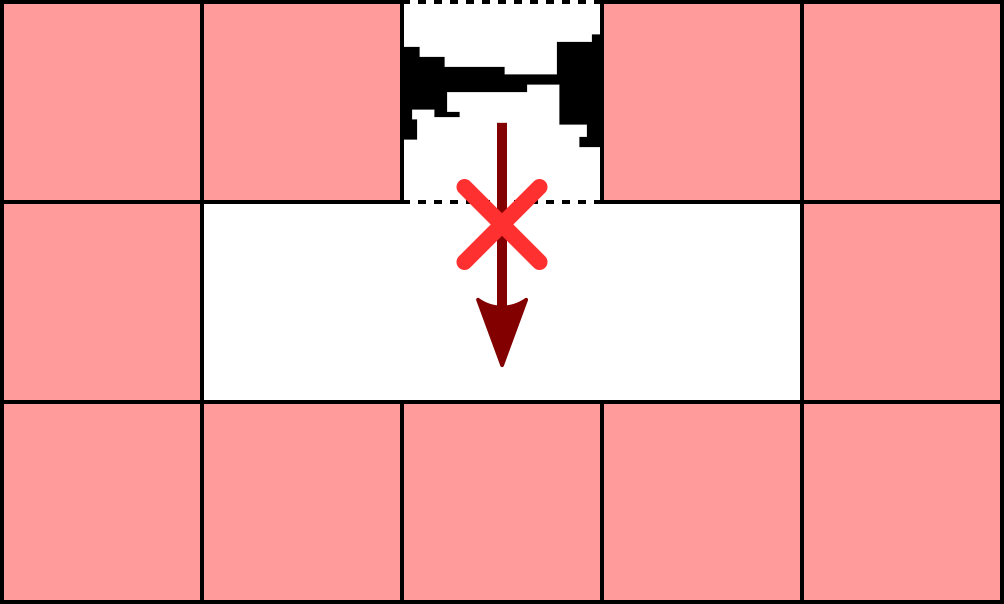}
    \caption{When checking for across-the-gap cooperation during a simulation, tiles can't naively span the entire gap without disconnecting two regions of space.}
    \label{fig:sim-atg}
\end{figure}

Initially defined in \cite{IUSA}, probes are simply tall rows of tiles placed at specific locations of a macrotile. Probes solve the problem illustrated in Figure~\ref{fig:sim-atg} of trying to simulate across-the-gap (AtG) cooperation. Naively, it would seem that it's necessary to block-off a region of space to learn if two opposite macrotiles could cooperate to place a tile in your location. We can avoid this by specifying regions where sequences of probes grow into locations corresponding to pairs of tiles which could potentially cooperate AtG. If AtG cooperation isn't possible between the relevant macrotiles, then the probes will be spaced sufficiently far to allow a data path to navigate through the probe region; otherwise, two probes will meet, with a gap of a single tile between them which is insufficient for a data path to fit. Instead a tile can cooperatively attach between the probes and initiate the growth of a new data path indicating that the current macrotile can resolve into the tile placed by the AtG cooperation.

\begin{figure}
    \centering
    \includegraphics[width=\textwidth]{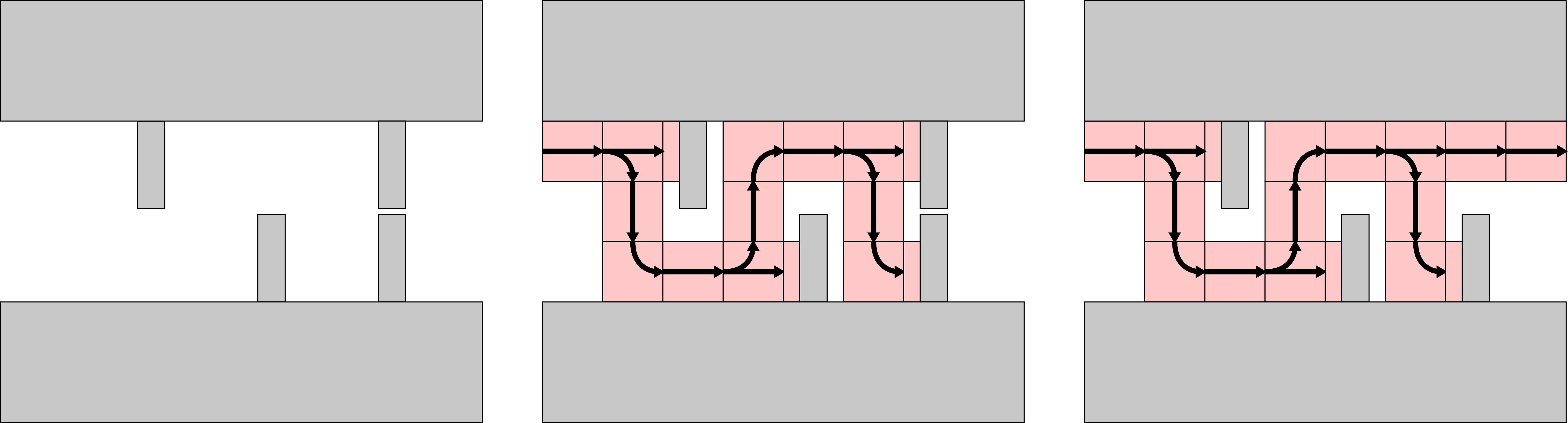}
    \caption{An illustration of probe regions. In the case that the probes do not meet, it will be possible for a TM blocks to grow around the probes. Otherwise, probes will block the growth of TM blocks and across-the-gap cooperation between the probes can initiate the growth of different TM blocks which perform a different task.}
    \label{fig:PaTAM-sim-DPaTAM-signal-through-gap}
\end{figure}

\subsubsection{Macrotile Protocol}

Given a DPaTAM system $\calT$, we describe the simulation of $\calT$ by some PaTAM system $\calS$ by describing the process by which tiles attach within macrotiles of $\calS$ to resolve into a tile in $\calT$. Growth within a macrotile location begins once at least one adjacent macrotile has resolved and our macrotiles must be robust to the arbitrary orders in which adjacent macrotiles can resolve. In this way, tile attachments within a macrotile behave similarly to a concurrent algorithm with special care needed to handle race conditions and avoid deadlocks when undesirable.

\paragraph*{Macrotiles and Component Blocks}

The fundamental component of used in our protocol are tile gadgets which we call \emph{component blocks} (CBs) which are made from several TM blocks and which may potentially initiate the growth of probe regions. Each macrotile in $\calS$ will consist of a 9x9 grid of CBs and in turn, each CB consists of a $2k\times 2k$ grid of TM blocks where $k$ is the number of ways of choosing $4$ tile glues from $\calT$. That is $k=g^4$ where $g$ is the number of glues appearing on tile types from $\calT$. Each of the TM blocks within a CB is large enough to store the following information:
\begin{itemize}
    \item a description of the simulated system $\calT$,
    \item a sequence of upto 4 glues from $\calT$,
    \item its position within the CB,
    \item the CBs position within a macrotile, and
    \item a description of the macrotile algorithm protocol which acts as the algorithm to be performed.
\end{itemize}
We will refer back to this list later to determine the scale factor of the simulation once the protocol has been described.

Figure~\ref{fig:sim-macrotile-schematic} illustrates the $9\times 9$ grid of component blocks including those which may contain probe regions. Each probe region may contain probes which grow from the adjacent CBs indicated by the directions of the arrows. The probes in a probe region are made from rows of TM blocks spanning half of the width of a CB though, one of the CBs which spawns a probe is selected beforehand to grow just shy of half-way through the probe region. This creates a single-tile wide gap between two aligned probes in which tiles can cooperatively attach to initiate the growth of other TM blocks which depend on the alignment of probes. Figure~\ref{fig:component-block-probes} illustrates two scenarios: one where two probes align in a probe region, in which case no path of TM blocks can pass through the region and in which cooperation between the probes determines the next CB, and a scenario where no probes align, in which case a path of TM blocks can pass to determine the next CB.

\begin{figure}
    \centering
    \includegraphics[width=0.7\textwidth]{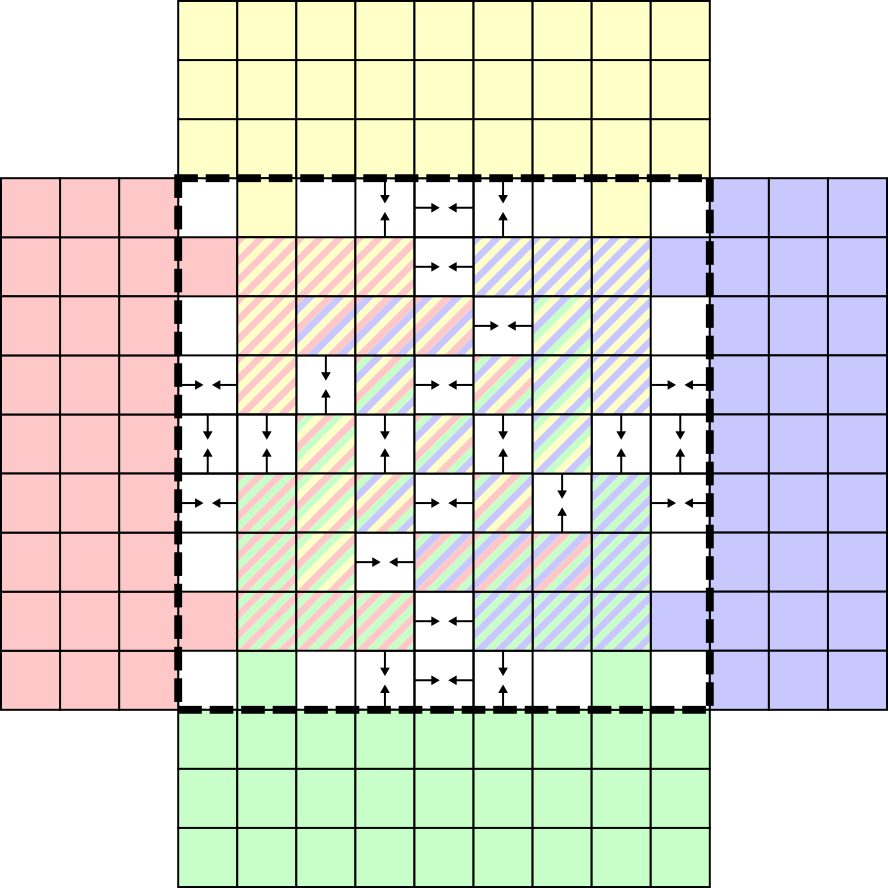}
    \caption{A schematic describing the $9\times 9$ grid of potential component blocks which may appear in a macrotile location. Squares containing two arrows indicate a grid location which may contain a probe region. Parts of adjacent macrotiles are indicated by the colored squares surrounding the macrotile. The colors of the grid locations within a macrotile describe which of the adjacent macrotiles each component block may have information from.}
    \label{fig:sim-macrotile-schematic}
\end{figure}

\begin{figure}
    \centering
    \includegraphics[width=0.9\textwidth]{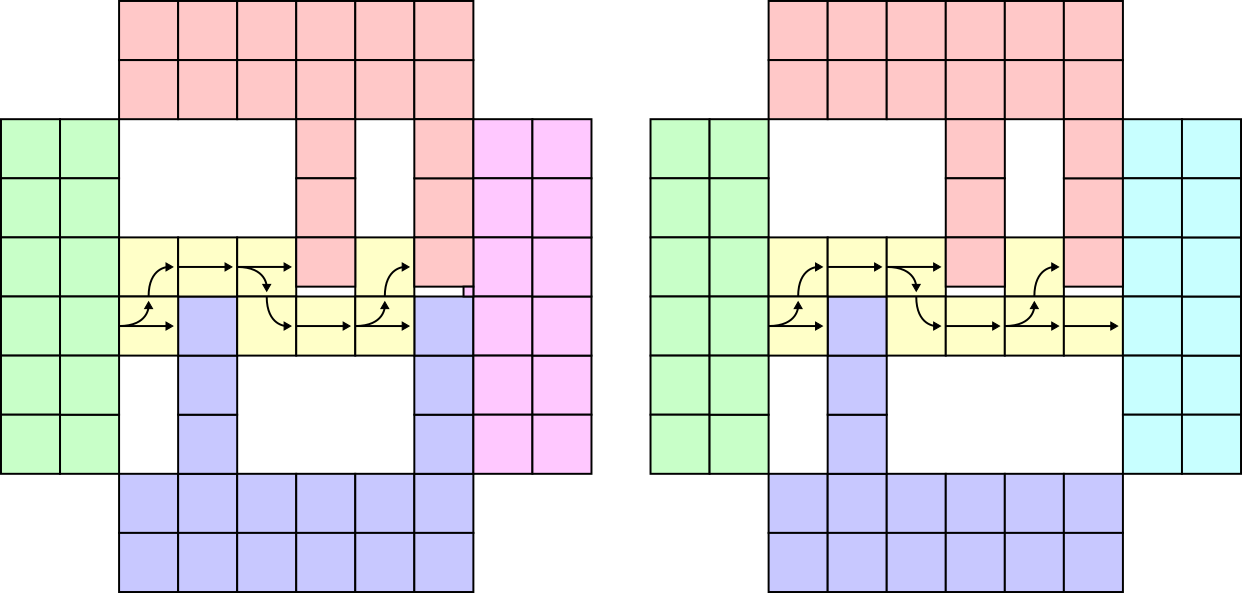}
    \caption{Probe regions between component blocks. The red and blue CBs present probes made from TM blocks into the grid location between them. The yellow TM blocks are spawned from those in the green CB and can only pass through the probe region if none of the probes align. In the case that the probes do align (left) the yellow TM blocks will not be able to pass and the CB on the right is initiated by cooperative growth between the probes. If the probes do not align (right), then the yellow TM blocks can pass through and initiate growth of the right CB. Therefore depending on the alignment of the probes, the CB to the right of the probe region is determined by either the CBs from which the probes grew or the CB on the left of the probe region.}
    \label{fig:component-block-probes}
\end{figure}

We can think of the component blocks as behaving similarly to individual tiles, with probe regions facilitating across-the-gap cooperation, and which perform the algorithm of simulating $\calT$ in the $9\times 9$ grid constituting a macrotile. These ``tiles'' grow along a path accumulating information regarding the surrounding macrotiles until they reach the center of the macrotile. Once there, by design of our protocol, the TM blocks making up the CB will have enough information to determine the tile from $\calT$ into which the macrotile should resolve. The specified probe regions indicate potential locations in which two CBs can ``consolidate'' information about the adjacent macrotiles while also allowing TM blocks to pass through in the case that probes don't align. Consequently, the size of the CBs, being $2k\times 2k$ TM blocks, was chosen so that there are always enough probe locations for every combination of glues presented by the adjacent macrotiles. The additional factor of 2 means that each probe has enough space for TM blocks to squeeze around even adjacent probes.

\paragraph*{Protocol for One Surrounding Macrotile}

\begin{figure}
    \centering
    \includegraphics[width=\textwidth]{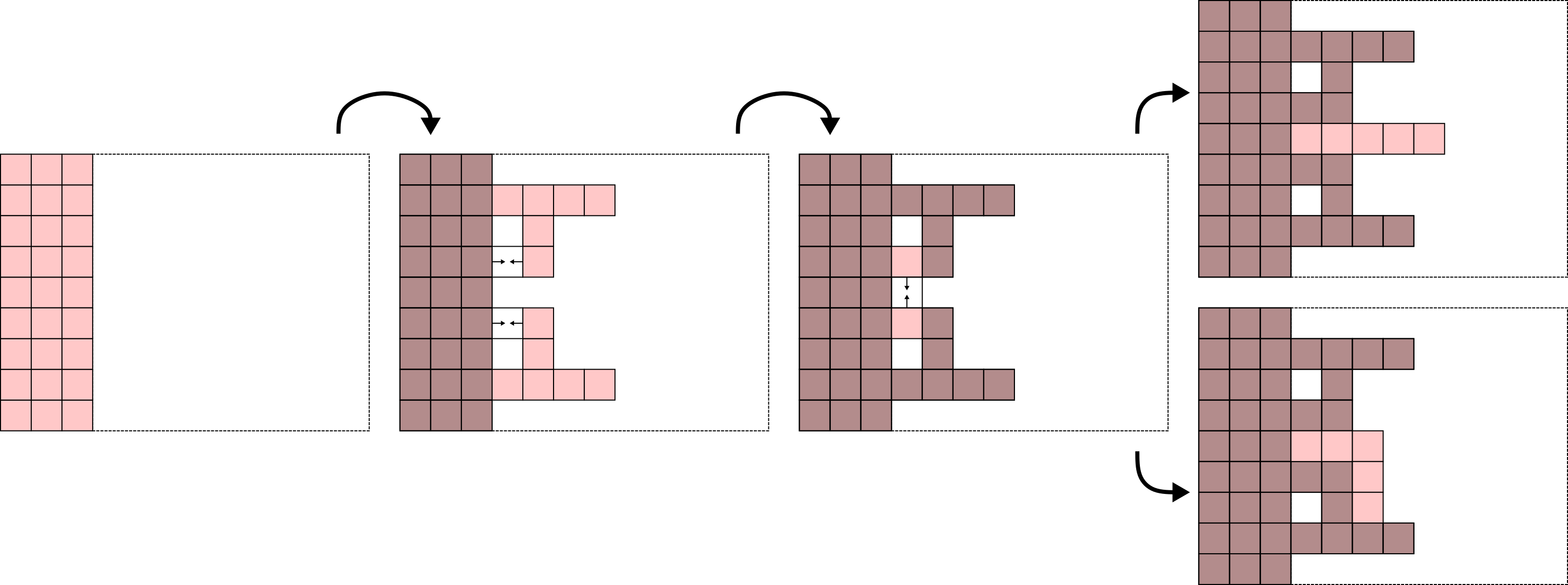}
    \caption{The pattern of CB growth from an adjacent macrotile. Starting from the adjacent macrotile, CBs will grow into mirrored \emph{hands} which cooperate to allow CBs to attach in the center of the first row in the macrotile. If this macrotile corresponds to a $\tau$-strength attachment in $\calT$, then a row of CBs will attach to the center (top) and will eventually resolve the macrotile, otherwise, they grow along a clockwise turn (bottom) to await growth from adjacent macrotiles.}
    \label{fig:sim-growth-pattern}
\end{figure}

We will describe the protocol which occurs in the $9\times 9$ grid of component blocks as though each CB is a tile. Cooperative binding of these ``tiles'' always occurs in a designated probe region and is facilitated by the probes, while $\tau$-strength attachments are facilitated by the TM blocks which make up each CB ``tile''. In other words, $\tau$-strength attachments of CBs is simply a fixed pattern of TM block growth which cause the CBs to grow along some predetermined path which may depend on the information stored in the TM blocks.

Illustrated in Figure~\ref{fig:sim-growth-pattern}, from an adjacent macrotile, CBs first grow into two mirrored ``T''-shaped structures called \emph{hands}. In the probe regions between these hands and the adjacent macrotile, CBs then grow cooperatively. Because, for the time being, we are only dealing with a single adjacent macrotile, these cooperative ``attachments'' just cause CBs to grow into the center of the first row of the macrotile. Then one of two things may happen. If the adjacent macrotile corresponds to a tile in $\calT$ whose glue can allow a $\tau$-strength attachment, then the TM blocks which make up the CBs has enough information to resolve the current macrotile. In this case, $\tau$-strength attachments cause the CBs to grow into the center of the macrotile and eventually cause the macrotile to resolve, as will be described later. Otherwise, a path of CBs will simply grow along the hand on the clockwise side of the macrotile and no further growth will occur until another adjacent macrotile resolves.

\paragraph*{Protocol for Two Adjacent Surrounding Macrotiles}

\begin{figure}[h!]
    \centering
    \includegraphics[width=\textwidth]{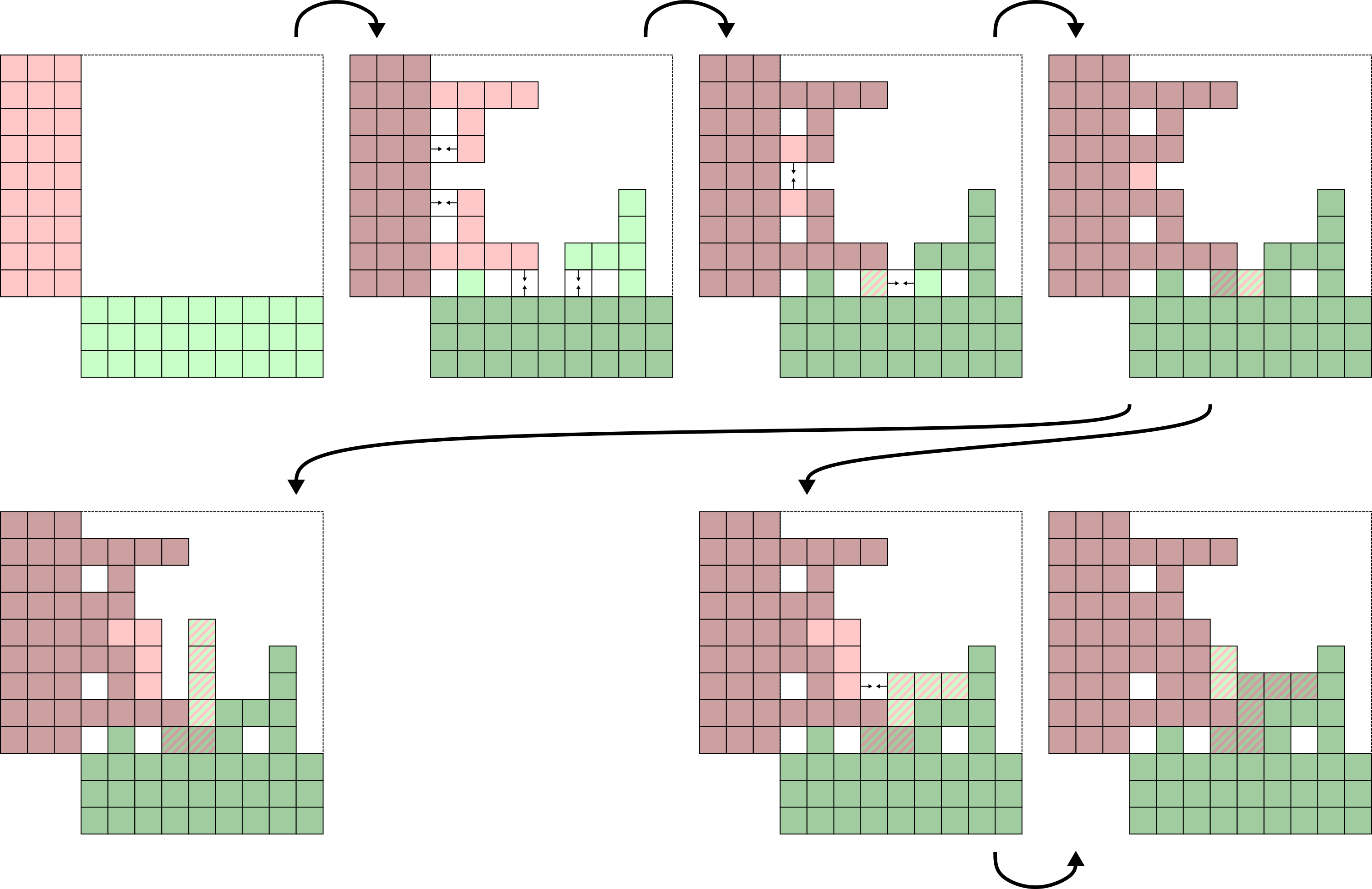}
    \caption{With two adjacent surrounding macrotiles, the protocol occurs similarly to before, but the two surround macrotiles will compete to place one of their hands.}
    \label{fig:sim-growth-two-adjacent}
\end{figure}

Figure~\ref{fig:sim-growth-two-adjacent} illustrates the protocol for two surrounding macrotiles which appear on adjacent sides of the current macrotile. We note that the surrounding macrotiles may resolve at different times and therefore our protocol is designed to handle different orders of CB attachment. Growth begins in much the same way as the case above with both surrounding macrotiles attempting to first grow their hands. Notice though that each will attempt to grow one of their hands in the same location. To handle this, we simply require that during the growth of the hands, the TM blocks which make up the CB blocks shared by both hands start by placing a single tile in the shared corner. Whichever hand is able to place that tile first is the one which grows its hand. In Figure~\ref{fig:sim-growth-two-adjacent} this happens to be the hand from the macrotile to the west, but it could have just as well been the one from the south.

Once the hands are grown, cooperative ``attachment'' of CBs occurs between the hands and the surrounding macrotiles just like before, but in this case one of the cooperative attachments will happen between one of the macrotiles and a hand which originated from the other macrotile. This cooperative CB attachment will therefore be able to consolidate the information from both adjacent macrotiles, storing both glues in the TM blocks which make the CB. If this information is enough to allow the macrotile to resolve, that is if the surrounding macrotiles represent tiles in $\calT$ which can cooperate to allow the attachment of a tile in this location, then CBs will ``attach'' with $\tau$-strength to the center of the macrotile which will eventually cause it to resolve. Otherwise, they both grow to form clockwise elbows just as in the case with only 1 surrounding macrotile. This time however, a probe region exists between the two adjacent elbows which will allow the further attachment of a few CBs. These attachments will not do anything until another surrounding macrotile resolves since no tile in $\calT$ can yet attach in the location corresponding to this macrotile.

Note that this process is robust to the order in which surrounding macrotiles resolves. Even if the west macrotile resolves well before the south and grows according to the protocol for a single surrounding macrotile, this protocol is still able to occur with the south macrotile being the one which loses the competition to place its hand. The situation is symmetrical so that it occurs similarly even if the south macrotile wins the competition. Additionally, even if the first surrounding macrotile to arrive corresponds to a tile in $\calT$ which is able to form a $\tau$-strength attachment in the current macrotile, the attachments formed by the other macrotile will not interfere.

\paragraph*{Protocol for Two Opposite Surrounding Macrotiles}

\begin{figure}[h!]
    \centering
    \includegraphics[width=\textwidth]{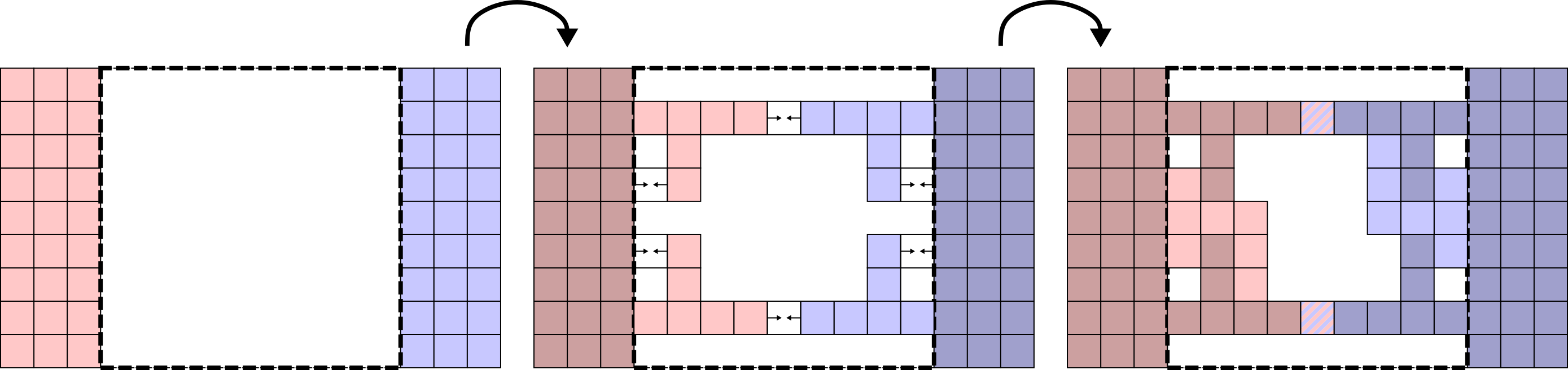}
    \caption{}
    \label{fig:sim-growth-two-opposite}
\end{figure}

In the case that the two surrounding macrotiles are opposite each other, the probe regions between the hands grown by both will facilitate the resolution of the macrotile as illustrated in Figure~\ref{fig:sim-growth-two-opposite}. Once the hands are grown, if the two opposing macrotiles correspond to tiles in $\calT$ capable of across-the-gap cooperation, then the probe regions between the opposing hands will have aligned probes. Using these probes, the TM blocks which make up the CB can consolidate the information from both surrounding macrotiles and consequently are able to determine the tile into which the current macrotile will resolve. Note that this case is slightly different from the previous cases since CBs will not grow into the center of the macrotile. In fact, the interior of the macrotile will become a constrained region into which no additional tiles can diffuse once both have grown. This doesn't matter however since both of the CBs between the hands have sufficient information to determine how the macrotile should resolve and all that needs to be done is have that information propagated to the empty sides of the tiles.

\begin{figure}[h!]
    \centering
    \includegraphics[width=\textwidth]{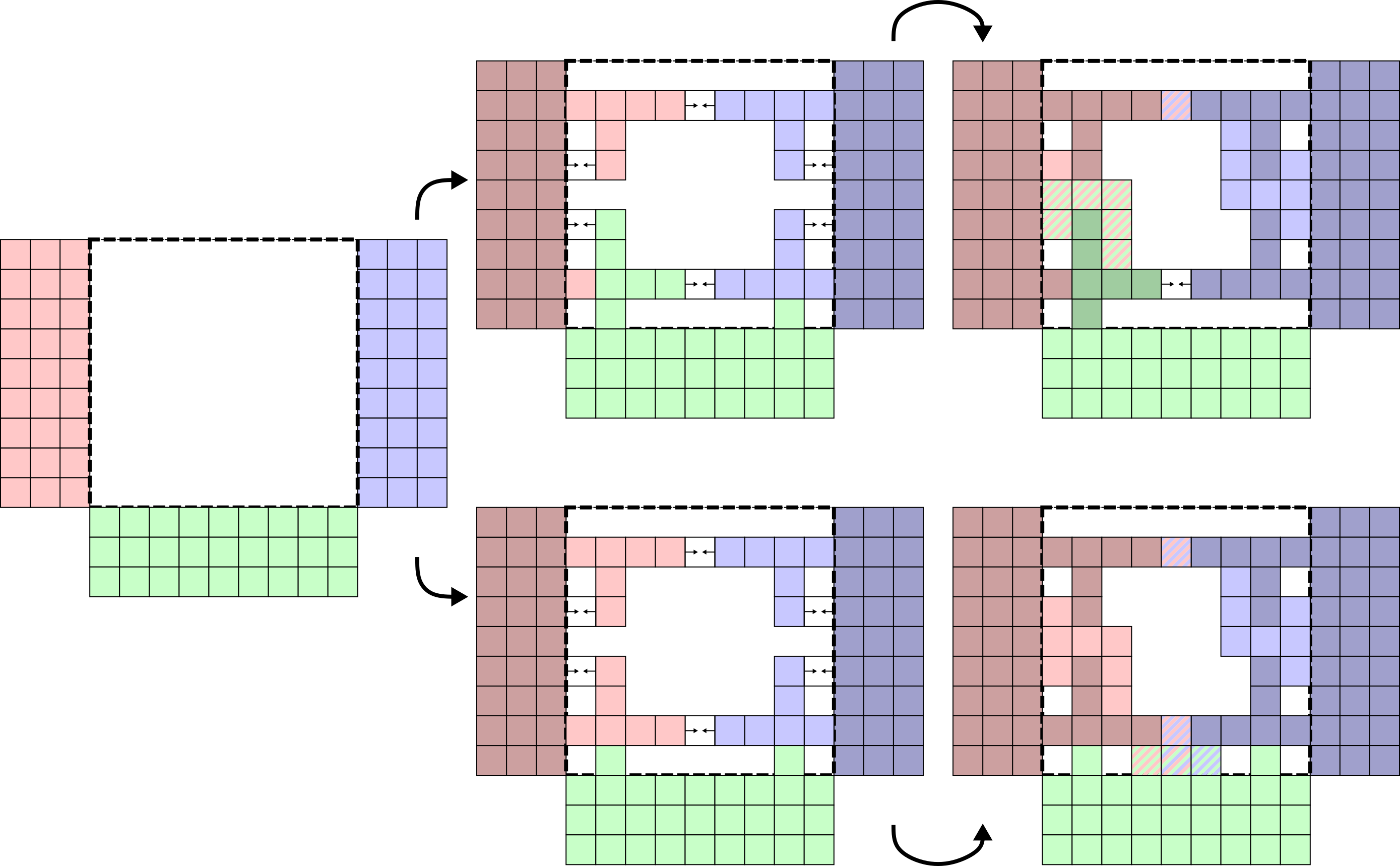}
    \caption{}
    \label{fig:sim-growth-two-opposite-one-extra}
\end{figure}

Even in the case that there is another surrounding macrotile between the opposing macrotiles, the protocol allows for the successful resolution of the macrotile as illustrated in Figure~\ref{fig:sim-growth-two-opposite-one-extra}. There are several possibilities regarding which macrotiles win the competition to place their hands, but these will not affect the probe region on the side which does not contain a surrounding macrotile. Consequently regardless of the tile attachments that occur on the south of Figure~\ref{fig:sim-growth-two-opposite-one-extra}, CBs can still attach on the north side of the macrotile indicating that it resolved into the tile corresponding to an across-the-gap cooperative attachment. Any growth below the north probe region between opposing hands will not be able to affect this resolution since it will be blocked by the probes and growth along the south side is irrelevant since the macrotile does not need to indicate its resolution in that direction as the surrounding macrotile to the south has already resolved. Also, since we are simulating a directed system, the current macrotile could not resolve into any other tile anyway.

If on the other hand, the two opposing macrotiles do not represent tiles in $\calT$ capable of across-the-gap cooperation, then no probes in the region shared between opposing hands will align and TM blocks will be able to grow in between, allowing the protocol to continue to other cases. In other words, opposing macrotiles which aren't capable of across-the-gap cooperation will not interfere with the continuation of the protocol. If an additional macrotile resolves to either the north or south and represents a tile in $\calT$ which can cooperate with either the east, west, or both surrounding macrotiles, it will still be able to do so.

\paragraph*{Protocol for Three Surrounding Macrotiles}

\begin{figure}[h!]
    \centering
    \includegraphics[width=0.9\textwidth]{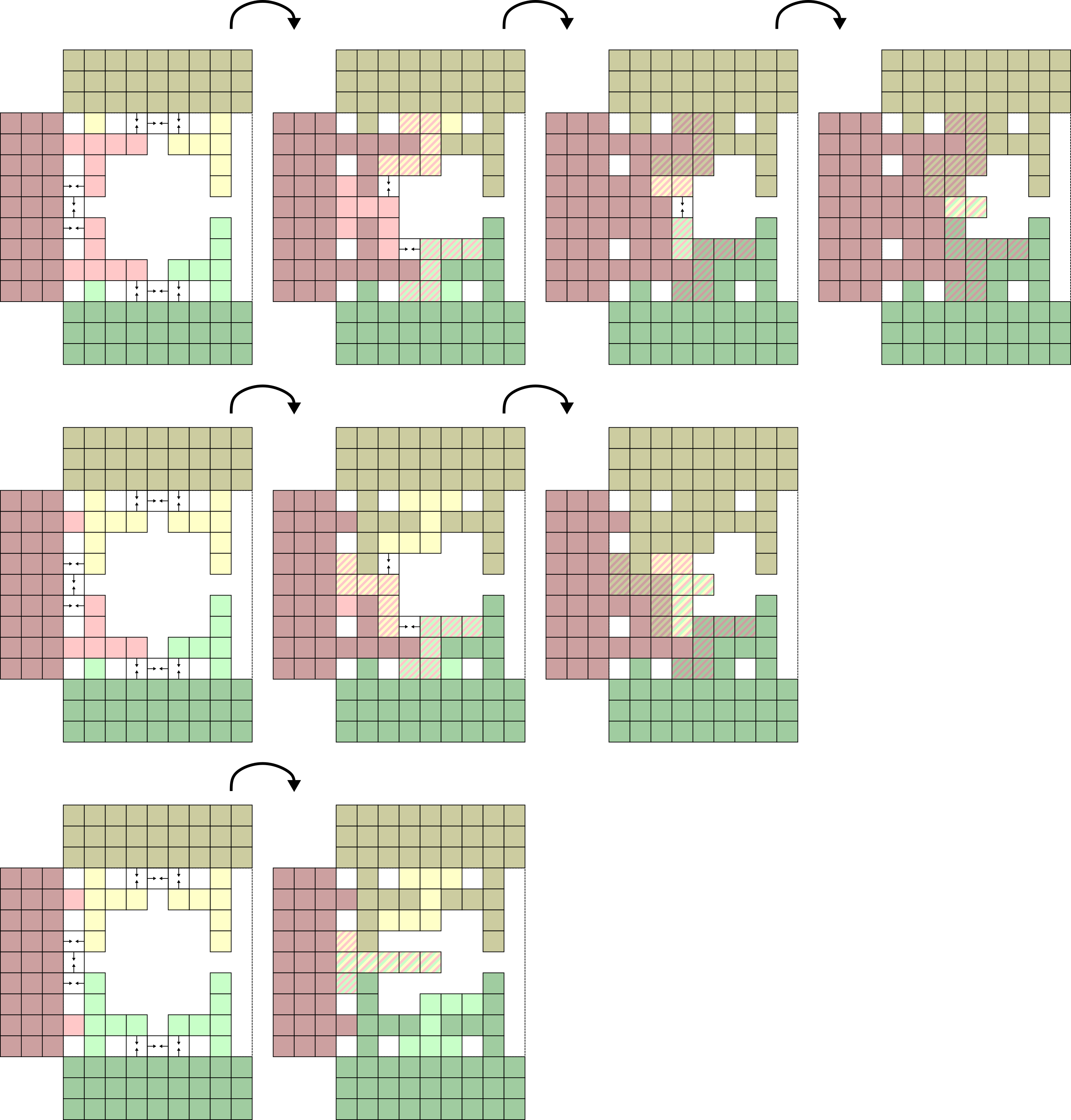}
    \caption{}
    \label{fig:sim-growth-three}
\end{figure}

In the case of three surrounding macrotiles which cooperate to resolve the macrotile, there are a few cases which need to be considered regarding which surrounding macrotiles win the race to place their hands. These cases are illustrated in Figure~\ref{fig:sim-growth-three}. The most complicated case occurs when the central surrounding macrotile wins both races and places both of its hands, though all cases are essentially the same. In this case, like before, all three sides grow their hands and corresponding clockwise elbows. These elbows then cooperate via probes to consolidate their information with CBs from the adjacent elbows. We are then left with CBs which can again cooperate to consolidate the information from all three sides and grow into the center of the macrotile to eventually resolve. The other cases are simpler in that the CBs consolidate the information of all three surrounding macrotiles earlier. In these cases, the CBs have enough information to ``know'' how the macrotile should resolve and are able to grow into the center earlier.

\paragraph*{Four Surrounding Macrotiles}

In our directed PaTAM system $\calT$, it is impossible for a tile to be placed in a location surrounded by tiles on all sides. Any tile location surrounded on all 4 sides is constrained; if all 4 tiles were required for the tile placement, it could not occur because of the diffusion restriction. If on the other hand, fewer than 4 tiles were required but it were still possible for 4 tiles to surround a location before the tile attached, the system would not be directed; the tile placement may or may not happen before the region is constrained. Such situations will never happen in our simulating system $\calS$ because they violate the assumption that $\calT$ is a DPaTAM system.

\paragraph*{Resolving a Macrotile}

\begin{figure}[h!]
    \centering
    \includegraphics[width=\textwidth]{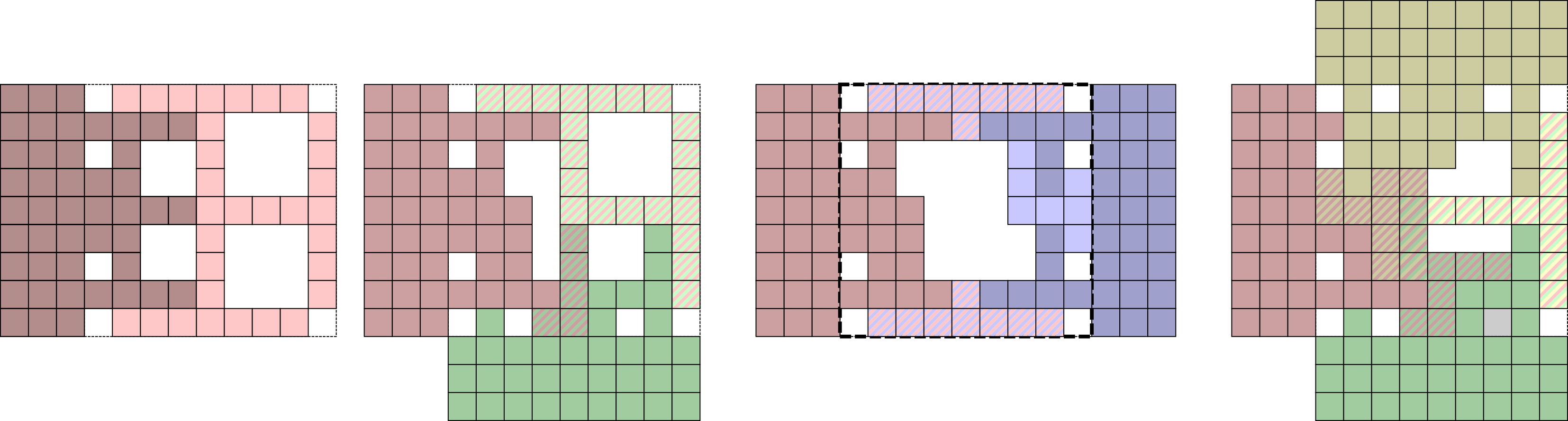}
    \caption{}
    \label{fig:sim-resolve}
\end{figure}

Once a CB reaches the center of a macrotile (or when across-the-gap cooperation happens in the probe regions between opposing macrotile hands), the macrotile is ready to resolve. Upon placement of the first tile in this center CB location, the macrotile representation function $R$ no longer maps the macrotile to empty space and instead maps to the correct tile in $\calT$ which has been determined by the CBs from the information from the necessary surrounding macrotiles. Once this happens, additional CBs begin ``attaching'' to grow in each of the 4 cardinal directions in order to grow a row of tiles to each necessary side of the macrotile as illustrated in Figure~\ref{fig:sim-resolve}. These rows of tiles then act to repeat the entire process in adjacent macrotile locations, growing hands and following the described protocol.

Note that it may be the case that CBs cannot fully grow in one of the directions. This could happen if one of the surrounding macrotiles resolves after much of the process has completed in the current macrotile for instance or if one of the directions requires tile placements in a constrained region. This isn't a problem however because only those surrounding macrotiles which could correspond to locations in $\calT$ where valid tile attachments can occur need to ``know'' about the current macrotiles resolved state.

\subsubsection{Scale Factor of the Simulation}

Recall that each TM block which makes up a component block contains the following information:
\begin{itemize}
    \item a description of the simulated system $\calT$,
    \item a sequence of upto 4 glues from $\calT$,
    \item its position within the CB,
    \item the CBs position within a macrotile, and
    \item a description of the macrotile algorithm protocol which acts as the algorithm to be performed.
\end{itemize}
For a given system $\calT=(T,\sigma,\tau)$, we note that only the tileset $T$ and binding threshold $\tau$ are necessary for the TM blocks and component blocks to determine how a macrotile should resolve. The seed $\sigma$ of $\calT$ is handled by the seed of $\calS$ since our simulation can begin with a seed assembly made of pre-resolved macrotiles already placed in the correct configuration to represent $\sigma$.

To describe a tileset $T$, we note that the number of glues appearing in $T$ is at most $4\sizeof{T}$. Each glue can therefore be encoded efficiently using $O(\log(\sizeof{T}))$ bits and consequently the entire tileset can be encoded using $O(\sizeof{T}\log(\sizeof{T}))$ bits. The binding threshold is just a positive integer and thus can be encoded using $O(\log(\tau))$ bits. For the purposes of our TM blocks, a table describing the tileset $T$ along with the binding threshold $\tau$ can thus be encoded using $O(\sizeof{T}\log(\sizeof{T}) + \log(\tau))$ bits. Additionally, a sequence of 4 glues simply requires $O(\log(\sizeof{T}))$ bits to encode.

The position of a TM block within a component block has coordinates between $0$ and $2k$ where $k$ is, again, the number of possible sequences of 4 glues. $k=g^4$ where $g$ is the number of glues so $\log(2k)=\log(2g^4)=O(\log(g))=O(\log(\sizeof{T}))$ and thus the position of a TM block can be encoded using $O(\log(\sizeof{T}))$ bits. Since the position of of a CB within a macrotile is always a point in $\{0,1,\ldots,8\}^2$ it simply requires a constant number of bits to store. Similarly, the algorithm performed by a TM block is fixed so it also requires only a constant number of bits to store. Therefore, the information stored by a TM block in our construction requires only $O(\sizeof{T}\log(\sizeof{T}))$ bits to encode.

Each TM block within a macrotile performs a fixed algorithm, implicitly described above. For the most part, this algorithm simply requires the TM blocks to initiate the growth of other TM blocks according to a fixed pattern to grow the current component block, but occasionally also requires querying the tileset $T$ to determine if the known glues are sufficient to induce a tile attachment in $\calT$. To grow in the fixed pattern, requires space and time on the order of the size of the CB which is $O(\sizeof{T}^4)$. To determine if the known glues are sufficient for an attachment in $\calT$, we can modify our naive encoding of the tileset $T$ and binding threshold $\calT$ with a slightly more complex encoding which consists of a table with entries for each combination of 4 glues. These entries simply store a reference to the tile type which may attach given the surrounding glues. Since we are simulating a directed system, we know that it should never be possible for multiple tiles to attach in a single location so no entry needs to be stored with more than one possible tile. Such an encoding requires $g^4$ entries each of size $O(\log(\sizeof{T}))$ since each holds just a single tile type's description. This table thus requires $O(\sizeof{T}^4\log(\sizeof{T}))$ space and can be queried in the same amount of time.

Overall, each TM block needs to therefore be $O(\sizeof{T}^4\log(\sizeof{T}))$ tiles large and since each CB is $2k = 2g^4$ TM blocks in each dimension, a component block must be $O(\sizeof{T}^8\log(\sizeof{T}))$ tiles wide in the worst case. Since there are a fixed number of CBs per macrotile, this is therefore the scale factor of the simulation. While the scale factor is polynomial in size, we note that it can be significantly improved by storing the tileset look-up table separately from the TM blocks, most of which do not need to query the table. In, previous iterations of our construction, this was done but required a significantly more complex protocol which was much more difficult to explain and understand. Regardless, the purpose of this construction is to demonstrate the existence of a universal tileset and show that the PaTAM is capable of all of the dynamics of its directed subset.

\bibliographystyle{plainurl}
\bibliography{tam,experimental_refs}

\end{document}